\documentclass[11pt, a4paper]{article}
\usepackage[utf8]{inputenc}
\usepackage[margin=25mm]{geometry}
\PassOptionsToPackage{usenames,dvipsnames,svgnames,table}{xcolor}
\usepackage[british]{babel}
\usepackage{fancyhdr}
\usepackage{amsmath}
\usepackage{amsfonts}
\usepackage{amssymb}
\usepackage{amsthm}
\usepackage{graphicx}
\usepackage{float}
\usepackage{titling}
\usepackage{centernot}
\usepackage{mathrsfs}
\usepackage{mathtools}
\usepackage{bbm}
\usepackage{halloweenmath}
\usepackage{listings}
\usepackage{diagbox}
\usepackage{multirow}
\usepackage{subcaption}
\usepackage{enumerate}
\usepackage{tikz}
\usepackage[backend=biber,style=alphabetic, maxbibnames=100]{biblatex}
\usepackage{csquotes}

\usepackage{footnote}
\makesavenoteenv{table}
\makesavenoteenv{tabular}

\usepackage{hyperref}
\hypersetup{%
pdfcreator = {},
pdfproducer = {},
pdfstartpage = {1},
pdfsubject = {},
pdftitle = {}}

\usepackage{xcolor}
\usepackage{pgfornament}

\author{Laura Man\v cinska, Thor Gabelgaard Nielsen,\thanks{Parts of this paper have been submitted as a Master's thesis at the University of Copenhagen by the second author.} and Jitendra Prakash}
\title{Glued magic games self-test maximally entangled states}
\date{\today}
\renewcommand\footnotemark{}

\newcommand{\numberhere}{\stepcounter{equation}\tag{\theequation}}

\newcommand{\bra}[1]{\left\langle{}#1\right|}
\newcommand{\ket}[1]{\left|{}#1\right\rangle}
\newcommand{\braket}[1]{\left\langle{}#1\right\rangle}

\newcommand\tensor\otimes
\newcommand\bigtensor\bigotimes
\newcommand{\squeeze}[2]{\bra{#1} #2 \ket{#1}}

\newcommand\dsum\oplus
\newcommand\bigdsum\bigoplus
\DeclareMathOperator\Ran{ran}

\newtheorem{theorem}{Theorem}

\newtheorem{corollary}[theorem]{Corollary}
\newtheorem{lemma}[theorem]{Lemma}
\theoremstyle{definition}
\newtheorem{definition}[theorem]{Definition}
\newtheorem{example}[theorem]{Example}

\numberwithin{theorem}{section}
\numberwithin{equation}{section}

\DeclareMathOperator{\Span}{span}
\DeclareMathOperator{\Supp}{supp}
\newcommand\aux{\mathit{aux}}

\newcommand\norm[1]{\ensuremath{\left\|{}#1\right\|}}

\newcommand\comm[1]{\ensuremath{\left[#1\right]}}

\usetikzlibrary{math}

\newcommand\markold[1]{
}
\newcommand\shinynew[1]{
#1
}

\begin{document}
\maketitle
\begin{abstract}
Self-testing results allow us to infer the underlying quantum mechanical description of states and measurements from classical outputs produced by non-communicating parties.
The standard definition of self-testing does not apply in situations when there are two or more inequivalent optimal strategies.
To address this, we introduce the notion of self-testing convex combinations of reference strategies, which is a generalisation of self-testing to multiple strategies.
We show that the Glued Magic Square game \cite{algebraic-chsh} self-tests a convex combination of two inequivalent strategies. As a corollary, we obtain that the Glued Magic square game self-tests two EPR pairs thus answering an open question from \cite{algebraic-chsh}.
Our self-test is robust and extends to natural generalisations of the Glued Magic Square game.
\end{abstract}

\section{Introduction}
In this paper, we consider the Glued Magic Square game which was originally constructed in \cite{algebraic-chsh} by building upon the Mermin-Peres Magic Square game \cite{mermin,peres}. The authors of \cite{algebraic-chsh} exhibit two inequivalent perfect quantum strategies for the Glued Magic Square game, thus showing that the game does not self-test any quantum strategy. As noted there, the two inequivalent quantum strategies still use equivalent states -- which is in fact unavoidable as we show that the Glued Magic Square game self-tests the maximally entangled state of local dimension four (Corollary \ref{cor:gms-state-selftest}). Furthermore we completely characterise the possible optimal strategies, but the usual self-testing definition is clearly insufficient, so we present a generalisation which allows us to self-test a convex combination of strategies. It should be noted that variations in the definition of self-testing has already appeared in the literature.

The Magic Pentagram game introduced by Mermin \cite{mermin} is very similar to the Magic Square game. It is then natural to consider replacing either one or both Magic Square parts in the Glued Magic Square game by Magic Pentagrams. While a construction similar to that of \cite{algebraic-chsh} still shows that these games do not self-test the measurements in any of their perfect quantum strategies, our arguments for the Glued Magic Square still goes through.

\subsection{Organisation of this paper}
In this first section we have briefly introduced the motivation for the Glued Magic Square game, and introduce the notation we will use in the rest of this paper.
In Section \ref{sec:preliminaries}, we formally define the nonlocal games we will consider including Glued Magic games, along with the relevant self-testing tools. In Section \ref{sec:self-testing}, we build up our some tools for proving self-testing and prove our main result. Furthermore we give examples of possible strategies for the Glued Magic Square game. In Section \ref{sec:robustness} we sketch how to translate our results into the robust case.

\subsection{Notation}
For a natural number $n$, we write $[n] := \{1, \dots, n\}$.  A \emph{quantum state} (or simply a \emph{state}) is a unit vector in some Hilbert space. Every state $\ket\psi \in \mathcal H_A \tensor \mathcal H_B$ admits a \emph{Schmidt decomposition}: $\ket\psi = \sum_i\lambda_i\ket{v_i}\ket{w_i}$ where $\lambda_i>0$ for all $i$, and $\{\ket{v_i}\}_i$ and $\{\ket{w_i}\}_i$ are orthonormal sets in $\mathcal H_A$ and $\mathcal H_B$, respectively. We write $\Supp_A\ket\psi$ for the subspace $\overline{\Span_i\{\ket{v_i}\}} \subseteq \mathcal H_A$, and similarly $\Supp_B\ket\psi := \overline{\Span_i\{\ket{w_i}\}} \subseteq \mathcal H_B$. We will also use $\Supp\ket\psi := \overline{\Span_i\{\ket{v_i}\ket{w_i}\}} \subseteq \mathcal{H}_A\otimes \mathcal{H}_B$. We say that $\ket{\psi}$ is of \emph{full Schmidt rank} if $\Supp_A\ket{\psi} = \mathcal H_A$ and $\Supp_B\ket{\psi} = \mathcal H_B$. For a natural number $k\geq 2$, we will write
\begin{equation*}
\ket{\psi_k} := \frac1{\sqrt k}\sum_{i=0}^{k-1} \ket{ii}
\end{equation*}
for the \emph{maximally entangled state} of local dimension $k$. By an \emph{observable}, we mean a self-adjoint unitary operator. Given an observable $A\in \mathcal{B}(\mathcal{H})$ we have the spectral decomposition $A = A^+ - A^-$, where the projection onto the $+1$-eigenspace (resp., $-1$-eigenspace) of $A$ is denoted by $A^+$ (resp., $A^-$). The commutator of two operators $A,B\in \mathcal{B}(\mathcal{H})$ is denoted by $[A,B] := AB-BA$. Let $\mathbb{Z}/d\mathbb{Z} := \{0,1,\dots,d-1\}$ denote the ring of integers modulo $d$ for a natural number $d$.

Suppose $\mathcal H$ has $n$ mutually orthogonal subspaces $\mathcal H_1, \dots, \mathcal H_n$. If for every $\ket v \in \mathcal H$ there exists a unique $\ket{v_k} \in \mathcal H_k$ such that $\ket v = \sum_{k=1}^n\ket{v_k}$, then $\mathcal H$ is a direct sum of the subspaces, and we write $\mathcal H = \bigdsum_{k=1}^n\mathcal H_k$. Unless stated otherwise, any direct sum of Hilbert spaces appearing in this paper will be an internal direct sum. Furthermore, suppose $A_k : \mathcal H_k \to \mathcal K_k$ is a bounded operator for each $k \in [n]$, and let $\mathcal K = \bigdsum_{k=1}^n \mathcal K_k$ be the external direct sum. By writing $\bigdsum_{k=1}^nA_k$, we mean the operator $A : \mathcal H \to \mathcal K$ acting on a vector $\ket v = \sum_{k=1}^n\ket{v_k}$, where $\ket{v_k} \in \mathcal H_k$ for each $k$, by
\begin{equation}
\label{eq:dsum-of-ops}
\left(\bigdsum_{k=1}^nA_k\right)\ket v = \bigdsum_{k=1}^nA_k\ket{v_k} = \left(A_1\ket{v_1}, A_2\ket{v_2}, \dots, A_n\ket{v_n}\right).
\end{equation}

We will often identify spaces consisting of tuples of vectors with other isomorphic Hilbert spaces, e.g., $\mathbb C^2 \times \mathbb C^2 \cong \mathbb C^4$.

\section{Preliminaries}\label{sec:preliminaries}

\begin{definition}[Nonlocal game]\label{def:nonlocal-game}
A \emph{nonlocal game} is a tuple $G = (X, Y, A, B, V, \pi)$, where $X, Y, A$ and $B$ are nonempty finite sets corresponding to the possible questions and answers for each prover, $V : X\times Y \times A \times B \to \{0,1\}$ is a predicate which we refer to as the \emph{verification function}, and $\pi$ is a probability distribution over $X \times Y$.
\end{definition}

A nonlocal game $G = (X, Y, A, B, V, \pi)$ is played by a referee and two provers (which we will denote by Alice and Bob). It proceeds by the referee choosing at random some pair of questions $(x,y) \in X \times Y$ according to the probability distribution $\pi$ and then sending $x$ to Alice and $y$ to Bob. Alice and Bob, who are not allowed to communicate in any way after they receive the questions, must then each choose an answer to send back, and upon receiving $a\in A$ from Alice and $b\in B$ from Bob, the referee evaluates $V(x,y,a,b)$ to determine whether the provers win. The provers win if $V(x,y,a,b) = 1$, otherwise they lose

We are interested in a particular subclass of nonlocal games, namely those based upon a system of linear constraints, and aptly called linear constraint system games. Such games have been studied in depth \cite{cm-characterisation}. Suppose we have a system of $n$ linear equations over $\mathbb Z/d\mathbb Z$ in $k$ variables. A natural question in this regard is whether the given system of equations is satisfiable, and this is easy to determine classically. This can also be done by constructing a nonlocal game where Alice is given an equation and has to fill in all the variables, while Bob is given a single variable which he is asked to fill in. The provers then win if and only if Alice's assignment satisfies the equation given to her and Bob's assignment to his variable is consistent with Alice's assignment. Formally, a linear constraint system game is defined as follows.

\begin{definition}[Linear constraint system game]
Let $n, k, d \in \mathbb N$, and consider a system of $n$ linear equations in $k$ variables $e_1,\dots,e_k$ over $\mathbb{Z}/d\mathbb{Z}$ given by
\begin{align}
\label{eq:sys-eq-1}
\sum_{j=1}^k\alpha_{i,j}e_j = \beta_i, \quad i\in [n],
\end{align}
where $\alpha_{i,j}, \beta_i \in \mathbb Z / d\mathbb Z$ for all $i \in [n]$ and $j \in [k]$. The \emph{linear constraint system game} (LCS game) corresponding to the system of linear equations in \eqref{eq:sys-eq-1} is a nonlocal game with question sets $X = [n],\, Y = [k],$ and answer sets $A = (\mathbb Z / d\mathbb Z)^k,\,  B = \mathbb Z / d\mathbb Z$. The probability distribution $\pi$ is the uniform distribution over the set $\{(x, y) \in X \times Y \mid \alpha_{x,y} \neq 0\}$, and the verification function is
\begin{equation*}
V(x,y,a,b) = \begin{cases}
1 & \text{if }\sum_{j=1}^k \alpha_{x,j}a_j = \beta_x \;\text{and}\; a_y = b, \\
0 & \text{otherwise}.
\end{cases}
\end{equation*}
Furthermore, if $d=2$, then the game is called a \emph{binary linear constraint system game}.
\end{definition}

Quantum strategies for nonlocal games are usually described by a state shared by Alice and Bob, and sets of measurement operators for both players which may be assumed to be projective. Equivalently, quantum strategies can be described in terms of unitaries which for our purposes will be more convenient. Furthermore, all the LCS games we consider are over $\mathbb Z/2\mathbb Z$, so observables will be sufficient.

\begin{definition}[Quantum strategy of observables]\label{def:quantum-strategy-observables}
Suppose $G = (X, Y, A, B, V, \pi)$ is a binary LCS game. A \emph{quantum strategy of observables} (or simply a \emph{quantum strategy}) is a 3-tuple:
\begin{align}\label{eq:quantum-strategy}
\mathcal S = \left(\ket\psi \in \mathcal H_A\tensor \mathcal H_B, \left\{A_i^{(x)}:i \in A\right\}_{x\in X},\left\{B_j\right\}_{j\in B}\right),
\end{align} where $\mathcal H_A$ and $\mathcal H_B$ are (possibly infinite dimensional) Hilbert spaces, $A_i^{(x)} \in \mathcal B(\mathcal H_A)$ and $B_j \in \mathcal B(\mathcal H_B)$ are observables, and $|\psi\rangle$ is a state. For question $x \in X$, Alice successively uses $A_1^{(x)}, \dots, A_k^{(x)}$ to assign values to each variable, while Bob only uses $B_j$ to determine his answer. Note that we associate projection onto the $+1$-eigenspace (resp., $-1$-eigenspace) with answer 0 (resp., answer 1).
\end{definition}

We will be using the term “perfect strategy” for any strategy which wins a nonlocal game with probability 1. Furthermore, we will be using the term “pseudo-telepathic nonlocal game” to refer to a nonlocal game with a perfect quantum strategy but no perfect classical strategies. Observe that it does not refer to the games which do not have any perfect quantum strategy but a limit of quantum strategies wins the game with probability 1. 

Observe that for a given quantum strategy as in Expression \eqref{eq:quantum-strategy}, it is not necessarily the case that Alice uses the same observables to assign values to the same variable within different equations she may be asked. That is, $A^{(x)}_i \neq A^{(x^{\prime})}_i$ may be possible for $x\neq x^{\prime}$. However, as shown in \cite{cm-characterisation}, if a quantum strategy is perfect for some binary LCS game, then one can show that equality holds in the previous inequality. Moreover, we recall the following theorem. 

\begin{theorem}\cite{cm-characterisation}\label{thm:operator-sol}
Let $G = (X, Y, A, B, V, \pi)$ be a binary LCS game and let $\mathcal{S}$ as given in Expression \eqref{eq:quantum-strategy} be a perfect quantum strategy for the game. If $\ket\psi$ is of full Schmidt rank, then \begin{enumerate}
\item $A_i^{(x)} = A_i^{(x')}$ for all $x, x'\in X$. By this token, we set $A_i := A_i^{(x)} = A_i^{(x')}$ for all $i\in Y$.
\item If variables $e_i$ and $e_j$ appear in the same equation, then $A_i$ and $A_j$ (resp., $B_i$ and $B_j$) commute.
\item The observables $\{A_i\}_{i\in Y}$ (resp., $\{B_j\}_{j\in Y}$) satisfy the linear equations when written in multiplicative form: $e_1^{\alpha_{x,1}}\dots e_k^{\alpha_{x,k}} = (-1)^{b_x}$ for all $x\in X$.
\end{enumerate}
\end{theorem}

As noted in \cite{cm-characterisation}, a system of linear constraints is satisfiable if and only if the corresponding nonlocal game can be won perfectly with some classical strategy. Therefore, insatisfiable systems of linear constraints will be of particular importance, as they may have a strictly better quantum strategy. It turns out that some of these even have perfect quantum strategies, a canonical example of which is the Magic Square game.

\begin{definition}[Magic Square game]
The \emph{Magic Square game} is the LCS game associated with the following system of equations over $\mathbb Z/2\mathbb Z$:
\begin{equation}
\label{eq:glued-magic-square-equations}
\begin{array}{ccccc}
e_1+e_2+e_3 = 0 && e_4+e_5+e_6 = 0 && e_7+e_8+e_9 = 0\\
e_1+e_4+e_7 = 0 && e_2+e_5+e_8 = 0 && e_3+e_6+e_9 = 1\\
\end{array}
\end{equation}
\end{definition}

It is also possible to visualize the system of equations in \eqref{eq:glued-magic-square-equations} as a graph -- Figure \ref{subfig:magic-square-game} -- from which the name derives. 

The Magic Pentagram game, which shares many properties with the Magic Square game, can be described by Figure \ref{subfig:magic-pentagram-game}, from which the relevant equations (over $\mathbb Z/2\mathbb Z$) can be readily derived.

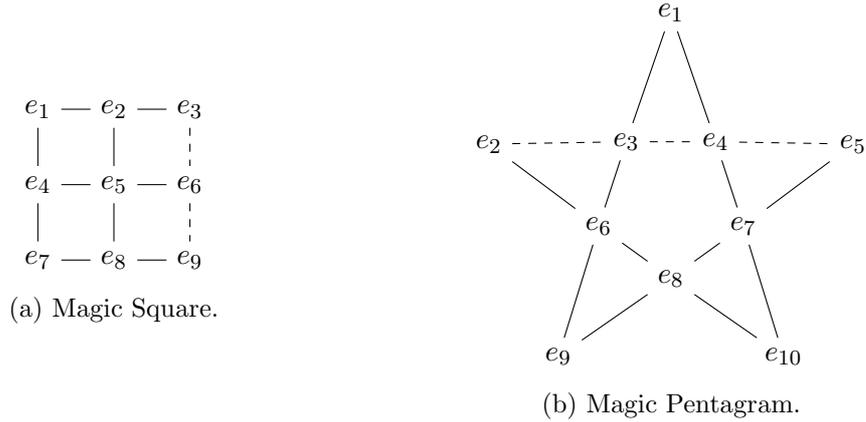
\begin{figure}
\centering
\begin{subfigure}{0.45\textwidth}
\centering
\begin{tikzpicture}
\node (e1) {$e_1$}; \node[right of=e1](e2) {$e_2$}; \node[right of=e2](e3) {$e_3$};
\node[below of=e1] (e4) {$e_4$}; \node[right of=e4](e5) {$e_5$}; \node[right of=e5](e6) {$e_6$};
\node[below of=e4] (e7) {$e_7$}; \node[right of=e7](e8) {$e_8$}; \node[right of=e8](e9) {$e_9$};
\draw (e1) -- (e2); \draw (e2) -- (e3);
\draw (e4) -- (e5); \draw (e5) -- (e6);
\draw (e7) -- (e8); \draw (e8) -- (e9);
\draw (e1) -- (e4); \draw (e4) -- (e7);
\draw (e2) -- (e5); \draw (e5) -- (e8);
\draw[dashed] (e3) -- (e6); \draw[dashed] (e6) -- (e9);
\end{tikzpicture}
\caption{Magic Square.}
\label{subfig:magic-square-game}
\end{subfigure}
\begin{subfigure}{0.45\textwidth}
\centering
\usetikzlibrary{math}
\begin{tikzpicture}
\node (e8) at ({cos(deg(-5*pi/10))}, {sin(deg(-5*pi/10))}) {$e_8$};
\node (e7) at ({cos(deg(-1*pi/10))}, {sin(deg(-1*pi/10))}) {$e_7$};
\node (e4) at ({cos(deg(3*pi/10))}, {sin(deg(3*pi/10))}) {$e_4$};
\node (e3) at ({cos(deg(7*pi/10))}, {sin(deg(7*pi/10))}) {$e_3$};
\node (e6) at ({cos(deg(11*pi/10))}, {sin(deg(11*pi/10))}) {$e_6$};
\tikzmath{\r = sin(deg(3*pi/10)) + 5^(1/3);}
\node (e1) at ({\r*cos(deg(5*pi/10))}, {\r*sin(deg(5*pi/10))}) {$e_1$};
\node (e2) at ({\r*cos(deg(9*pi/10))}, {\r*sin(deg(9*pi/10))}) {$e_2$};
\node (e9) at ({\r*cos(deg(13*pi/10))}, {\r*sin(deg(13*pi/10))}) {$e_9$};
\node (e10) at ({\r*cos(deg(17*pi/10))}, {\r*sin(deg(17*pi/10))}) {$e_{10}$};
\node (e5) at ({\r*cos(deg(21*pi/10))}, {\r*sin(deg(21*pi/10))}) {$e_5$};

\draw (e1) -- (e3); \draw (e1) -- (e4);
\draw (e3) -- (e6); \draw (e4) -- (e7);
\draw (e6) -- (e9); \draw (e7) -- (e10);
\draw (e2) -- (e6); \draw (e5) -- (e7);
\draw (e6) -- (e8); \draw (e7) -- (e8);
\draw (e8) -- (e10); \draw (e8) -- (e9);
\draw[dashed] (e2) -- (e3);
\draw[dashed] (e3) -- (e4);
\draw[dashed] (e4) -- (e5);
\end{tikzpicture}
\caption{Magic Pentagram.}
\label{subfig:magic-pentagram-game}
\end{subfigure}
\caption{The two canonical examples of magic games. A straight line through several variables indicates that there is an equation where the variables should multiply to $+1$ if the line is solid, and $-1$ if it is dashed.}
\label{fig:magic-games}
\end{figure}

This now brings us to the main type of games we will be considering, namely \emph{glued games}, which are straightforward generalisations of the Glued Magic Square game considered in \cite{algebraic-chsh}.

\begin{definition}[Glued games]
Suppose $G$ is a binary LCS game in the variables $e_1, \dots, e_k$ and $H$ is a binary LCS game in the variables $e_{k+1}, \dots, e_{k+\ell}$. Suppose furthermore they each have exactly one equation in which the variables sum to  1. The \emph{glued game} corresponding to these two games is the binary LCS game in the variables $e_1, \dots, e_{k+\ell}$ whose equations are the union of all equations in $G$ and $H$, except for the two equations having variables sum up to 1 -- which we replace by a single equation by adding their left-hand sides and equating it to 1.
\end{definition}

The main example of such a game is the Glued Magic Square game which is the gluing of the Magic Square game with itself. The resulting game is illustrated in Figure \ref{fig:gms-visual}.
Gluing the Magic Square together with another Magic Square, or a Magic Pentagram, yields a pseudo-telepathy game, as one possible perfect quantum strategy is to use a perfect quantum strategy for the Magic Square game on the Magic Square part, and use the identity operators on the other part (for example, see Section 9 of \cite{algebraic-chsh}). Similarly, gluing together the Magic Pentagram with itself yields a pseudo-telepathy game.
\begin{figure}
\centering
\begin{tikzpicture}
\node (e1) {$e_1$}; \node[right of=e1](e2) {$e_2$}; \node[right of=e2](e3) {$e_3$};
\node[below of=e1] (e4) {$e_4$}; \node[right of=e4](e5) {$e_5$}; \node[right of=e5](e6) {$e_6$};
\node[below of=e4] (e7) {$e_7$}; \node[right of=e7](e8) {$e_8$}; \node[right of=e8](e9) {$e_9$};
\draw (e1) -- (e2); \draw (e2) -- (e3);
\draw (e4) -- (e5); \draw (e5) -- (e6);
\draw (e7) -- (e8); \draw (e8) -- (e9);
\draw (e1) -- (e4); \draw (e4) -- (e7);
\draw (e2) -- (e5); \draw (e5) -- (e8);
\draw[dashed] (e3) -- (e6); \draw[dashed] (e6) -- (e9);
\node[below of=e9] (e10) {$e_{10}$}; \node[right of=e10] (e11) {$e_{11}$}; \node[right of=e11] (e12) {$e_{12}$};
\node[below of=e10] (e13) {$e_{13}$}; \node[right of=e13] (e14) {$e_{14}$}; \node[right of=e14] (e15) {$e_{15}$};
\node[below of=e13] (e16) {$e_{16}$}; \node[right of=e16] (e17) {$e_{17}$}; \node[right of=e17] (e18) {$e_{18}$};
\draw[dashed] (e9) -- (e10);
\draw[dashed] (e10) -- (e13); \draw[dashed] (e13) -- (e16);
\draw (e11) -- (e14); \draw (e14) -- (e17);
\draw (e12) -- (e15); \draw (e15) -- (e18);
\draw (e10) -- (e11); \draw (e11) -- (e12);
\draw (e13) -- (e14); \draw (e14) -- (e15);
\draw (e16) -- (e17); \draw (e17) -- (e18);
\end{tikzpicture}
\caption{A visualisation of the Glued Magic Square game, a LCS game with 18 variables. In rows and columns, a solid line indicates that the variables of that row or column multiply to $+1$, while a dashed line indicates that the product instead should be $-1$.}
\label{fig:gms-visual}
\end{figure}
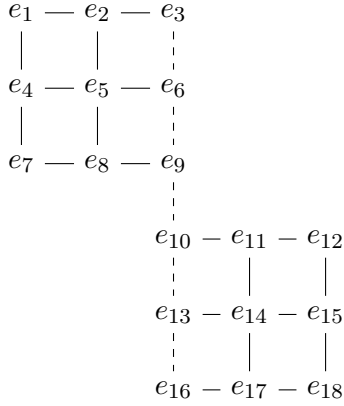

We now formally introduce the notion of self-testing. We recall the definition of a local dilation from \cite{MPS}.

\begin{definition}[Local dilation]
\label{def:local-dilation}
\sloppy{}Suppose $\mathcal S = \left(\ket\psi \in \mathcal H_A \tensor \mathcal H_B, \{A_i\}_i, \{B_j\}_j\right)$ and $\tilde{\mathcal S} = \left(\ket{\tilde\psi} \in \tilde{\mathcal H}_A \tensor \tilde{\mathcal H}_B, \{\tilde A_i\}_i, \{\tilde B_j\}_j\right)$ are two quantum strategies having the same number of observables for each party. We say that $\tilde{\mathcal S}$ is \emph{a local dilation of $\mathcal S$} if there exist Hilbert spaces $\mathcal H_{A,\aux}$ and $\mathcal H_{B,\aux}$, a state $\ket\aux \in \mathcal H_{A,\aux} \tensor \mathcal H_{B,\aux}$ and isometries $U_A : \mathcal H_A \to \tilde{\mathcal H}_A \tensor \mathcal H_{A,\aux}$ and $U_B : \mathcal H_A \to \tilde{\mathcal H}_B \tensor \mathcal H_{B, aux}$ such that with $U := U_A\tensor U_B$ it holds that for all $i$ and $j$,
\begin{align*}
U \ket\psi &= \ket{\tilde\psi}\tensor\ket\aux, \numberhere\label{eq:local-dilation-state-def}\\
U\left(A_i\tensor I\right) \ket\psi &= \left(\tilde{A_i}\tensor I\right)\ket{\tilde\psi}\tensor\ket\aux, \\
U\left(I\tensor B_j\right) \ket\psi &= \left(I\tensor\tilde{B_j}\right)\ket{\tilde\psi}\tensor\ket\aux. \numberhere\label{eq:local-dilation-measurement-def}
\end{align*}
\end{definition}

\begin{definition}[Self-testing]\label{def:usual-self-test}
A nonlocal game $G=(X, Y, A, B, V, \pi)$ is a \emph{self-test} for a quantum strategy (which we term \emph{ideal} or \emph{canonical}) $\tilde{\mathcal{S}}$, if for any quantum strategy $\mathcal S$ achieving the quantum value of $G$, the strategy $\tilde{\mathcal{S}}$ is a local dilation of $\mathcal{S}$. If only Equation \eqref{eq:local-dilation-state-def} is satisfied, we say that the game $G$ self-tests the state.
\end{definition}

We remark that in some cases it is also possible to show that the self-testing is \emph{robust}, i.e., for all strategies close to optimal, the measurements and state used will also be close to an optimal one in a certain sense, though this is beyond the scope of the main text of this paper.

It turns out that self-testing holds for both the Magic Square and the Magic Pentagram games. We will state a very simplified version of the results -- it should be noted that both games in fact \emph{robustly} self-test their ideal quantum strategies.

\begin{theorem}[Self-testing of Magic Square, \cite{wbms}]
\label{thm:magic-square-selftest}
The Magic Square game self-tests its ideal strategy $\mathcal S_{\text{MS}}$ which consists of nine observables for each party along with the state $\ket{\psi_4} = \frac12\sum_{i=0}^3 \ket{ii}$.
\end{theorem}

\begin{theorem}[Self-testing of Magic Pentagram, \cite{magic-pentagram-rigidity}]
\label{thm:magic-pentagram-selftest}
The Magic Pentagram game self-tests its ideal strategy $\mathcal S_{\text{MP}}$ which consists of ten observables for each party along with the state $\ket{\psi_8} = \frac1{\sqrt8}\sum_{i=0}^7 \ket{ii}$.
\end{theorem}

Later in the paper, we will be working with several different perfect quantum strategies for the same game, and therefore, to be able to completely characterise its perfect quantum strategies, we will need to discuss convex combinations of quantum strategies.

\begin{definition}[Convex combination of quantum strategies]
\label{def:strategy-linear-comb}
For each $k\in [n]$, let
\begin{align*}
\mathcal S_k = \left(\ket{\psi^{(k)}},\left\{A_i^{(k)}\right\}_i, \left\{B_j^{(k)}\right\}_j \right)
\end{align*}
be a quantum strategy for some nonlocal game $G$, and $\{\alpha_k\}_{k=1}^n \subseteq [0,1]$ a set of scalars satisfying $\sum_{k=1}^n\alpha_k^2 = 1$. The \emph{corresponding convex combination} is the quantum strategy
\begin{align*}
\mathcal S := \sum_{k=1}^n \alpha_k\mathcal S_k :=   \left(\ket\psi, \left\{A_i\right\}_i, \left\{B_j\right\}_j\right), 
\end{align*}
where the state and the observables are given by
\begin{align*}
\ket\psi = \bigdsum_{k=1}^n \alpha_k\ket{\psi^{(k)}}, && A_i = \bigdsum_{k=1}^n A_i^{(k)}, && B_j = \bigdsum_{k=1}^n B_j^{(k)}.
\end{align*}
Observe that the direct sums appearing here can be interpreted in two different ways. Either one can consider the external direct sum, where one forms a new Hilbert space consisting of ordered tuples of elements from each of the original spaces. In this case, we will refer to the convex combination as an \emph{external convex combination}. Another possible interpretation is that each of the strategies $\mathcal S_k$ lives in some subspace of some larger Hilbert space, and the direct sums are therefore the internal direct sum. In this case we refer to the convex combination as an \emph{internal convex combination}.
\end{definition}

Note that the direct sum of states appearing in the definition above is actually a direct sum of tensor product spaces, and for $\ket{\psi^{k}} \in \mathcal H_A^k \tensor \mathcal H_B^k$, the resulting state is thus $\ket\psi \in \bigdsum_{k,\ell=1}^n \left(\mathcal H_A^k\tensor \mathcal H_B^\ell\right)$, which is orthogonal to $\mathcal H_A^k\tensor \mathcal H_B^\ell$ for $k\neq \ell$ (in the sense that projecting the state onto each of these subspaces yields the zero vector). As an example with the external direct sum, it holds that $\frac{\sqrt2}{\sqrt5}\ket{\psi_2} \dsum \frac{\sqrt3}{\sqrt5}\ket{\psi_3} = \ket{\psi_5}$.

At first it might not be clear why we distinguish between internal and external direct sums, as they are isomorphic. However, internal direct sums have a desirable property that they are stable under unitaries, i.e., if one fixes some strategy for some game which decomposes into an internal direct sum, then applying a unitary to that strategy yields a strategy which again decomposes into an internal direct sum. However, that is not the case for external direct sums.

With the notion of a convex combination of quantum strategies in hand, we are ready to generalize the usual definition of self-testing given above by introducing the notion of convex self-testing where we allow multiple inequivalent optimal quantum strategies.

\begin{definition}[Convex self-testing]
\label{def:convex-self-test}
A nonlocal game $G=(X,Y,A,B,V,\pi)$ is a \emph{convex self-test} for the ideal quantum strategies $\tilde{\mathcal S}_1, \dots, \tilde{\mathcal S}_n$, if for any quantum strategy $\mathcal S$ achieving the quantum value of $G$, there exist coefficients $\{\alpha_k\}_{k=1}^n \subseteq [0, 1]$ and a decomposition of $\mathcal S$ into an internal convex combination $\mathcal S = \sum_{k=1}^n \alpha_k\mathcal S_k$ such that each $\tilde{\mathcal S}_k$ is a local dilation of $\mathcal S_k$.
\end{definition}

At first glance, while this appears to be quite different from the usual definition of self-testing, that is not the case, as we can rewrite it with more familiar conditions.

\begin{lemma}
Let $G$ be a nonlocal game and let $\tilde{\mathcal S}_k = \left(\ket{\tilde\psi^{(k)}}\in \tilde{\mathcal H}_A^k \tensor \tilde{\mathcal H}_B^k, \{\tilde A_i^{(k)}\}_i, \{\tilde B_j^{(k)}\}_j\right)$ be optimal quantum strategies for the game for each $k\in [n]$. Then the following are equivalent.
\begin{enumerate}
\item\label{enu:convex-equivalence-1} The game $G$ is a convex self-test for $\tilde{\mathcal S}_1, \dots, \tilde{\mathcal S}_n$.
\item\label{enu:convex-equivalence-2} Let $\mathcal S = \left(\ket\psi \in \mathcal H_A \tensor \mathcal H_B, \{A_i\}_i, \{B_j\}_j\right)$ be an optimal quantum strategy of $G$. Then for each $k\in [n]$, there exist Hilbert spaces $\mathcal H_{A,\aux}^k, \mathcal H_{B,\aux}^k$, states $\ket{aux_k}\in \mathcal H_{A,\aux}^k\otimes \mathcal H_{B,\aux}^k$, and isometries $U_A : \mathcal H_A \to \bigdsum_{k=1}^n \left(\tilde{\mathcal H}_A^k \tensor \mathcal H_{A,\aux}^k\right)$, $U_B : \mathcal H_B \to \bigdsum_{k=1}^n \left(\tilde{\mathcal H}_B^k \tensor \mathcal H_{B,\aux}^k\right)$ such that defining $U := U_A \tensor U_B$, it holds that:
\begin{align*}
U \ket\psi &= \bigdsum_{k=1}^n\left(\alpha_k\ket{\tilde\psi^{(k)}}\tensor\ket{\aux_k}\right),\numberhere\label{eq:conv-self-test-state}\\
U\left(A_i\tensor I\right) \ket\psi &= \bigdsum_{k=1}^n\left(\alpha_k\left(\tilde{A_i}\tensor I\right)\ket{\tilde\psi^{(k)}}\tensor\ket{\aux_k}\right)\numberhere\label{eq:conv-self-test-measurement-alice}\\
U\left(I\tensor B_j\right) \ket\psi &= \bigdsum_{k=1}^n\left(\alpha_k\left(I\tensor\tilde{B_j}\right)\ket{\tilde\psi^{(k)}}\tensor\ket{\aux_k}\right)\numberhere\label{eq:conv-self-test-measurement-bob}
\end{align*}
\end{enumerate}
\end{lemma}

\begin{proof}
First we assume that \eqref{enu:convex-equivalence-1} holds. Let $\mathcal S$ be an optimal quantum strategy for the game as in \eqref{enu:convex-equivalence-2}. By Definition \ref{def:convex-self-test}, there exist scalars $\{\alpha_k\}_{k=1}^n\subseteq [0,1]$ such that $\mathcal S = \sum_{k=1}^n \alpha_k\mathcal S_k$, where $\mathcal S_k = \left(\ket{\psi^{(k)}}\in {\mathcal H}_A^k \tensor {\mathcal H}_B^k, \{A_i^{(k)}\}_i, \{B_j^{(k)}\}_j\right)$, and that for each $k \in [n]$, $\tilde{\mathcal S}_k$ is a local dilation of $\mathcal S_k$. By Definition \ref{def:local-dilation}, there exist isometries $U_A^{(k)} : \mathcal H_A^k \to \tilde{\mathcal H}_A^k \tensor \mathcal H_{A,\aux}^k$ and $U_B^{(k)} : \mathcal H_B^k \to \tilde{\mathcal H}_B^k \tensor \mathcal H_{B,\aux}^k$ and an auxiliary state $\ket{\aux_k} \in \mathcal H_{A,\aux}^k \tensor \mathcal H_{B,\aux}^k$ such that the equations in Definition \ref{def:local-dilation} hold.

Define $U_A := \bigdsum_{k=1}^nU_A^{(k)}$ and similarly $U_B := \bigdsum_{k=1}^nU_B^{(k)}$. Observe that while these are defined as the direct sum of operators, this is done as described in Equation \eqref{eq:dsum-of-ops}, so assuming all of the operators are finite-dimensional, the matrix representations of $U_A$ and $U_B$ are not necessarily block-diagonal.

Now we find that
\begin{align*}
(U_A \tensor U_B)\ket\psi &= \left(\bigdsum_{\ell=1}^nU_A^{(\ell)}\right) \tensor \left(\bigdsum_{m=1}^nU_B^{(m)}\right)\left(\bigdsum_{k=1}^n\alpha_k\ket{\psi^{(k)}}\right)\\
&= \left(\bigdsum_{\ell,m=1}^n U_A^{(\ell)}\tensor U_B^{(m)}\right)\left(\bigdsum_{k=1}^n\alpha_k\ket{\psi^{(k)}}\right)\\
&= \bigdsum_{k=1}^n\left(\alpha_k\left(U_A^{(k)}\tensor U_B^{(k)}\right)\ket{\psi^{(k)}}\right)\\
&= \bigdsum_{k=1}^n\left(\alpha_k\ket{\tilde{\psi}^{(k)}}\tensor \ket{\aux_k}\right), \numberhere\label{eq:conv-self-test-state-eq}
\end{align*}
where we in the third equality used that $\ket{\psi_k} \in \mathcal H_A^k \tensor \mathcal H_B^k$ for all $k\in[n]$.

A similar computation reveals that
\begin{align*}
(U_A \tensor U_B)(A_i \tensor I)\ket\psi &= \bigdsum_{k=1}^n\left(\alpha_k\left(\tilde A_i^{(k)}\tensor I\right)\ket{\tilde{\psi}^{(k)}}\tensor \ket{\aux_k}\right), \numberhere\label{eq:conv-self-test-alice-op}\\
(U_A \tensor U_B)(I \tensor B_j)\ket\psi &= \bigdsum_{k=1}^n\left(\alpha_k\left(I \tensor \tilde B_j^{(k)}\right)\ket{\tilde{\psi}^{(k)}}\tensor \ket{\aux_k}\right).\numberhere\label{eq:conv-self-test-bob-op}
\end{align*}

While these equations are much closer to those usually found in self-testing statements, there is a crucial difference (apart from the direct sum): In the usual self-testing statements there is a single auxiliary state, but here one appears for each substrategy, and in general they can be different and thus  cannot be factored out from the expression above.

We now show the converse; namely that if \eqref{enu:convex-equivalence-2} holds 
(with the direct sums assumed to be internal direct sums), then so does \eqref{enu:convex-equivalence-1}. 
Note that we need only show that we can decompose the original strategy into $n$ orthogonal substrategies. 
For each $k\in[n]$ define $\mathcal H_A^k := \Supp_A\left(U_A^*\left(\ket{\tilde{\psi}^{(k)}} \tensor\ket{\aux_k}\right)\right)$ and similarly $\mathcal H_B^k := \Supp_B\left(U_B^*\left(\ket{\tilde{\psi}^{(k)}} \tensor\ket{\aux_k}\right)\right)$. Observe that we need only show that $\mathcal H_A^k \perp \mathcal H_A^\ell$ for all $k \neq \ell$ to show that $\mathcal H_A = \bigdsum_{k=1}^n \mathcal H_A^k$, and from that it follows that Definition \ref{def:convex-self-test} holds.

Note that $U_A$ is an isometry and obviously $\ket{\tilde{\psi}^{(k)}}\tensor \ket{\aux_k} \in \Ran(U_A) \tensor \tilde{\mathcal H}_B \tensor \mathcal H_{B,\aux}$ for all $k \in [n]$. As $U_A^*$ is an isometry on $\Ran(U_A)$, and $\left\{\ket{\tilde{\psi}^{(k)}}\tensor \ket{\aux_k}\right\}_{k=1}^n$ are all mutually orthogonal, this implies that applying $U_A^*\tensor I$ to each of these yield $n$ orthogonal vectors, and furthermore it holds that $\ket\psi = \sum_{k=1}^n (U_A^*\tensor I)\left(\ket{\tilde{\psi}^{(k)}}\tensor \ket{\aux_k}\right)$. However, then $\Supp_A\left((U_A^*\tensor I)\left(\ket{\tilde{\psi}^{(k)}}\tensor \ket{\aux_k}\right)\right)$ is orthogonal to $\Supp_A\left((U_A^*\tensor I)\left(\ket{\tilde{\psi}^{(\ell)}}\tensor \ket{\aux_\ell}\right)\right)$ for every $\ell \neq k$, which was what we needed to show.
\end{proof}

\section{Convex self-testing of the Glued Magic Square game}
\label{sec:self-testing}
We now show that if two self-tested pseudo-telepathy binary LCS games are glued together to form another pseudo-telepathy game, then the resulting glued game self-tests the common state of the two subgames, as long as the self-tested quantum strategies for each subgame satisfy certain commutation relations. We show that these relations hold for the Magic Square. For the Magic Pentagram such relations also follow similarly and are relegated to Appendix \ref{sec:pentagram}.
First, we need a lemma to show that if some operator preserves the state, then it acts like the identity operator.

\begin{lemma}
\label{lem:state-preserve-implies-identity}
Suppose $\ket\psi \in \mathcal H_A\tensor \mathcal H_B$ is a state and $G \in \mathcal B(\mathcal H_A)$ is some operator satisfying $(G\tensor I_B)\ket\psi = \ket\psi$. Then $G|_{\Supp_A(\ket\psi)} = I_A|_{\Supp_A(\ket\psi)}$.
\end{lemma}

\begin{proof}
Let $\ket\psi = \sum_i \lambda_i \ket{v_i,w_i}$ be a Schmidt decomposition with  $\lambda_i > 0$, and orthonormal sets $\{\ket{v_i}\} \subseteq \mathcal H_A$ and $\{\ket{w_i}\} \subseteq \mathcal H_B$. Since $(G\tensor I_B)\ket\psi = \ket\psi$, we get:
\begin{equation*}
\sum_i\lambda_i(G\ket{v_i})\ket{w_i} = \sum_i\lambda_i\ket{v_i}\ket{w_i},
\end{equation*}
but left-multiplying with $\bra{w_i}$ we get by orthonormality of $\{\ket{w_i}\}$ that $\lambda_i(G\ket{v_i}) = \lambda_i\ket{v_i}$, so $G\ket{v_i} = \ket{v_i}$. Thus, as $\{\ket{v_i}\}_i$ constitutes a basis for $\Supp_A(\ket\psi)$, we get the desired result.
\end{proof}

Before we are able to use this, we will need a result allowing us to restrict observables to the support of the state.

\begin{lemma}
Suppose $A \in \mathcal B(\mathcal H_A)$ and $B \in \mathcal B(\mathcal H_B)$ are two observables, and $\ket\psi \in \mathcal H_A \tensor \mathcal H_B$ is a state such that $\squeeze\psi{(A \tensor B)} = 1$. Then by defining $\mathcal H_A^1 := \Supp_A(\ket\psi)$ and $\mathcal H_A^0 := \left(\mathcal H_A^1\right)^\perp$, it holds that $\Ran(A|_{\mathcal H_A^k}) \subseteq \mathcal H_A^k$ for $k\in\{0,1\}$, and furthermore $A = A|_{\mathcal H_A^0} \dsum A|_{\mathcal H_A^1}$.
\end{lemma}

\begin{proof}
Initially observe that $\norm{A\tensor B} \leq 1$, and so Cauchy-Schwarz implies that $A\tensor B\ket\psi = \ket\psi$. Now Schmidt decompose the state $\ket\psi = \sum_i\lambda_i \ket{v_i}\ket{w_i}$, and observe that $\mathcal H_A^1 = \overline{\Span_i\{\ket{v_i}\}}$. Suppose for contradiction that $\Ran(A|_{\mathcal H_A^1}) \cap \mathcal H_A^0 \neq \{0\}$. Then for some $j$, as $\{\ket{v_i}\}$ constitutes a basis for $\mathcal H_A^1$, $A\ket{v_j} = \ket{\varphi_0} + \ket{\varphi_1}$, where $\ket{\varphi_k} \in \mathcal H_A^k$, for $k \in \{0, 1\}$, and $\ket{\varphi_0} \neq 0$. However, it therefore holds that
\begin{align*}
\sum_i \lambda_i \ket{v_i}\ket{w_i} &= \ket\psi = (A \tensor B)\ket\psi\\
&= \sum_i\lambda_i(A\ket{v_i}) \tensor(B\ket{w_i})\\
&= \sum_{i \neq j}\lambda_i(A\ket{v_i})\tensor (B\ket{w_i}) + \lambda_j\ket{\varphi_1}\tensor (B\ket{w_j})+\lambda_j\ket{\varphi_0} \tensor (B\ket{w_j})
\end{align*}
But noting that $A$ and $B$ are both observables, they are therefore isometries and thus as $\{\ket{v_i}\}_i$ and $\{\ket{w_i}\}_i$ are both linearly independent sets, all vectors in the above sum are linearly independent. But as $\ket{\varphi_0} \notin \mathcal H_A^1$, this is a contradiction as $\ket\psi \in \mathcal H_A^1 \tensor \mathcal H_B^1$. Therefore $\Ran(A|_{\mathcal H_A^1}) \subseteq \mathcal H_A^1$, and then by taking $\ket\phi \in \mathcal H_A^1$ we observe that $A^{-1} \ket\phi = A\ket\phi \in \mathcal H_A^1$, so it necessarily also holds that $\Ran(A|_{\mathcal H_A^0}) \subseteq \mathcal H_A^0$.
\end{proof}

In case of a game with consistent observables, such as the Magic Square game or the Magic Pentagram game, this implies that it is well-defined to restrict Alice's (resp., Bob's) operators to $\Supp_A(\ket\psi)$ (resp., $\Supp_B(\ket\psi)$. Therefore, we can without loss of generality assume the state to have full Schmidt rank.

\begin{lemma}
\label{lem:glue-ms-commutation}
Fix a perfect strategy $\left(\ket\psi \in \mathcal H_A\tensor \mathcal H_B, \{A_x\}_{x}, \{B_y\}_{y}\right)$ for a glued pseudo-telepathy binary LCS game with a Magic Square part, and suppose the observables $A_1, \dots, A_9$ for Alice (resp., $B_1, \dots, B_9$ for Bob) correspond to the Magic Square part, as illustrated in Figure \ref{subfig:magic-square-game}. Then with $E := A_3A_6A_9$ and $F:=B_3B_6B_9$, it holds that
\begin{align*}
\comm{E, A_i}\tensor I_B\ket\psi = 0, \quad \text{ and } \quad I_A \tensor \comm{F, B_j}\ket\psi = 0
\end{align*}
for all $i,j \in [9]$.
\end{lemma}

\begin{proof}
Note that by symmetry, we need only show $\left(\vphantom{\big|}\comm{E, A_i}\tensor I_B\right)\ket\psi = 0$ for all $i \in [9]$. Note furthermore that in any perfect strategy, $A_3$, $A_6$ and $A_9$ commute pairwise (Theorem \ref{thm:operator-sol}(b)), so in particular, they also commute with $E$. Therefore we need only show the statement for $i \in \{1, 2, 4, 5, 7, 8\}$. 

Note that we can Schmidt decompose the state as $\ket\psi = \sum_i \lambda_i \ket{v_i,w_i}$, for $\lambda_i > 0$ and orthonormal sets $\{\ket{v_i}\} \subseteq \mathcal H_A$, and $\{\ket{w_i}\} \subseteq \mathcal H_B$. Initially suppose the state has full Schmidt rank, which in particular implies that $\{\ket{v_i}\}$ constitutes an orthonormal basis for $\mathcal H_A$. Then, to prove that an operator is the identity operator, it suffices to note that it preserves $\ket\psi$ and then use Lemma~\ref{lem:state-preserve-implies-identity}.

This implies that any product of operators in a row or column of the Magic Square game (as illustrated in Figure \ref{subfig:magic-square-game}) is identity, except for $E$, as these products applied to the state must preserve the state.  We can now use this to prove the required commutation relations. Observe for example that we therefore have $A_1A_2A_3 = I$ and $A_4A_5A_6 = I$, and by using that these are all $\pm1$-valued observables, commuting within each equation, we get $A_3 = A_2A_1$, $A_6 = A_4A_5$, and thus
\begin{equation*}
A_3A_6 = A_2A_1A_4A_5 = A_2A_7A_5,
\end{equation*}
by using that $A_1A_4A_7 = I$, i.e. $A_1A_4 = A_7$. Therefore,
\begin{equation*}
A_3A_6A_8 = A_2A_7A_5A_8 = A_2A_7A_2,
\end{equation*}
by using that $A_2A_5A_8 = I$. But this latter fact now implies, using commutation of $A_3$ with $A_6$ and self-adjointness of the observables, that:
\begin{equation*}
A_3A_6A_8 = A_2^*A_7^*A_2^* = (A_2A_7A_2)^* = (A_3A_6A_8)^* = A_8^*A_6^*A_3^* = A_8A_6A_3.
\end{equation*}
Thus, $\comm{A_3A_6,A_8} = 0$. A similar argument gives $\comm{A_3A_6,A_7} = 0$. However, as $A_9$ commutes with $A_3$, $A_6$, $A_7$ and $A_8$, we clearly get $\comm{E,A_7} = \comm{A_3A_6A_9, A_7} = 0$, and also $\comm{E, A_8} = 0$. Commutation for $i \in \{1,2,4,5\}$ can be shown similarly.

Now we need only consider the case where $\ket\psi$ does not have full Schmidt rank. In this case, note that for $i \in \{1, \dots, 18\}$, we have $\left(A_i \tensor B_i\right) \ket\psi = \ket\psi$, since $\squeeze\psi{\left(\vphantom{\big|}A_i\tensor B_i\right)} = 1$ and $\norm{A_i \tensor B_i} \leq 1$. However this implies $\Span\{\ket{v_i}\}$ is invariant under all $A_j$, and similarly $\Span\{\ket{w_i}\}$ is invariant under all $B_j$. Therefore we can just restrict the operators to these subspaces, and then we from the above part get the desired statement.
\end{proof}

This lemma along with Lemma \ref{lem:glue-mp-commutation} show that if we glue together the Magic Square or the Magic Pentagram with themselves or each other, then it is well-defined to restrict a perfect quantum strategy for the glued game to one of its constituent parts. For example, if we consider the Glued Magic Square game, then we can restrict a perfect strategy for Alice to the $-1$-eigenspace of $E$, at least if we only do that for $A_1, \dots, A_9$. If we do so, then on that subspace $A_3A_6A_9 = -I$, and this implies that the restricted operators, which we again denote by  $A_1, \dots, A_9$, constitute a perfect strategy for the Magic Square game. While this cannot be used to extract the original observables as we have restricted to a particular space (and it was shown explicitly in \cite{algebraic-chsh} that self-testing does not hold, so that is impossible to do in general), it turns out that the self-testing properties of Magic Square force the state to be equivalent to $\ket{\psi_4}$ up to local isometry. Thus, we end up with self-testing  the state, which we formally prove in Theorem \ref{thm:gms-convex-self-test}.

We now consider the set of ideal strategies for the Glued Magic Square game.

\begin{example}
\label{ex:gms-irreducible-strategies}
Let $\tilde A_1, \dots, \tilde A_9$ denote the optimal operators (for Alice) for the Magic Square game from Theorem \ref{thm:magic-square-selftest}. Let $\sigma : (\mathbb Z/2\mathbb Z)^{\times 4} \to \mathcal B(\mathcal H)$ be a representation of the abelian group $(\mathbb Z/2\mathbb Z)^{\times 4} = \langle e_1, e_2, e_3, e_4 \mid e_i^2 = \comm{e_i, e_j} = 1\;\forall i,j\in[4]\rangle$. Set
\begin{align*}
G_1 &:= \sigma(e_1) & G_2 &:= \sigma(e_2) & G_3 &:= \sigma(e_1e_2)\\
G_4 &:= \sigma(e_3) & G_5 &:= \sigma(e_4) & G_6 &:= \sigma(e_3e_4)\\
G_7 &:= \sigma(e_1e_3) & G_8 &:= \sigma(e_2e_4) & G_9 &:= \sigma(e_1e_2e_3e_4)
\numberhere\label{eq:representation-mapping}
\end{align*} Define $A_i := B_i := \tilde A_i$ and $A_{19-i} := B_{19-i} := G_i$ for $i \in [9]$. Then the strategy $\mathcal S_1^\sigma := \left(\ket{\psi_4}, \left\{A_i\right\}_{i\in[18]}, \left\{B_j\right\}_{j\in[18]}\right)$ constitutes a perfect strategy for the Glued Magic Square game, assuming that $G_k$ commutes with $\tilde A_\ell$ for each $k,\ell\in\{3,6,9\}$. This can easily be verified by observing that the first nine observables for each of Alice and Bob constitute a perfect strategy for one Magic Square part, while their remaining observables all commute, square to identity and when taken the product within any row or column of how they are written in \eqref{eq:representation-mapping}, they multiply to identity.

Symmetrically, one can swap $A_i$ with $A_{i+9}$ and $B_i$ with $B_{i+9}$ for all $i\in [9]$ to get another strategy, which we will call $\mathcal S_2^\sigma$.

Note that our choice of subscript in the above example is not arbitrary, but instead indicates on which part the strategy is the optimal Magic Square strategy, i.e., a subscript of $1$ indicates that we use this for the first nine observables, corresponding to the top part of Figure \ref{fig:gms-visual}, while a subscript of $2$ indicates that we instead use the Magic Square strategy on the bottom part.
\end{example}

We now state our main theorem which characterises all the possible optimal strategies for the Glued Magic Square game. Note that while this is close to being a convex self-test in the sense of Definition \ref{def:convex-self-test}, it is not entirely so, as it in fact self-tests a family of strategies.

\begin{theorem}
\label{thm:gms-convex-self-test}
Let $\hat{\mathcal S}_k = \left\{\mathcal S_k^\sigma \mid \text{$\sigma$ is a representation of $(\mathbb Z/2\mathbb Z)^{\times 4}$}\right\}$ be the two families of possible strategies for the Glued Magic Square game as given in Example \ref{ex:gms-irreducible-strategies}.
The Glued Magic Square game is a convex self-test of the families $\hat{\mathcal S}_1$ and $\hat{\mathcal S}_2$ in the sense that for any perfect strategy $\mathcal S$ for the Glued Magic Square game, there exists representations $\sigma_1, \sigma_2$ of $(\mathbb Z/2\mathbb Z)^{\times 4}$, strategies $\mathcal S_1, \mathcal S_2$ and constants $\alpha_1, \alpha_2 \in [0,1]$, such that $\mathcal S$ is an internal convex combination of $\mathcal S_1$ and $\mathcal S_2$, i.e. $\mathcal S = \alpha_1\mathcal S_1 + \alpha_2\mathcal S_2$, and $\mathcal S_k^{\sigma_k} \in \hat{\mathcal S}_k$ is a local dilation of $\mathcal S_k$ for each $k\in[2]$. 
\end{theorem}
\begin{proof}
Let $\mathcal{S} = (\ket\psi\in \mathcal{H}_A\otimes \mathcal{H}_B,\{A_i\}_{i\in [18]},\{B_j\}_{j\in [18]})$ be a perfect quantum strategy for the Glued Magic Square game. Set
\begin{align*}
E=A_3A_6A_9, \qquad  F=A_{10}A_{13}A_{16}, \qquad G := B_3B_6B_9, \qquad H := B_{10}B_{13}B_{16}.
\end{align*}

We begin by decomposing the spaces $\mathcal{H}_A$ and $\mathcal{H}_B$. Consider the eigendecomposition: $E = E^+-E^-$ and $F = F^+-F^-$. As the quantum strategy $\mathcal{S}$ is assumed to be perfect, we have $(EF\tensor I_B)\ket\psi = -\ket\psi$, which implies
\begin{align*}
-1 = \squeeze\psi{\left(\vphantom{\big|}EF \tensor I_B\right)} = \squeeze\psi{\left(\vphantom{\big|}(EF)^+ \tensor I_B\right)} - \squeeze\psi{\left(\vphantom{\big|}(EF)^- \tensor I_B\right)},\numberhere\label{eq:gms-projector-split}
\end{align*}
and since $(EF)^+$ and $(EF)^-$ are projections, we necessarily have $0 \leq \squeeze\psi{\left(\vphantom{\big|}(EF)^+\tensor I_B\right)} \leq 1$ and $0 \leq \squeeze\psi{\left(\vphantom{\big|}(EF)^-\tensor I_B\right)} \leq 1$. Therefore, the only possible way in which Equation \eqref{eq:gms-projector-split} can hold is if $\squeeze\psi{(EF)^+\tensor I_B} = 0$ and $\squeeze\psi{(EF)^-\tensor I_B} = 1$, the latter of which  implies $\left((EF)^- \tensor I_B\right) \ket\psi = \ket\psi$. By using $I_A = E^++E^-$ and $I_A = F^+ + F^-$, we get:
\begin{align*}
(EF)^- = E^+F^-+E^-F^+ = (I_A-E^-)F^-+E^-(I_A-F^-) = E^-+F^--2E^-F^-
\end{align*}
However, as $(EF)^+ = E^+F^+ + E^-F^-$, Equation \eqref{eq:gms-projector-split} implies that
\begin{equation}
\label{eq:gms-proof-ef-is-identity}
\ket\psi =  \left((EF)^-\tensor I_B\right)\ket\psi = \left((E^-+F^-)\tensor I_B \right)\ket\psi.
\end{equation}

Let $\ket\psi = \sum_i \lambda_i \ket{v_i}\ket{w_i}$ be a  Schmidt decomposition with $\lambda_i>0$ and  orthonormal sets $\left\{\ket{v_i}\right\}_i \subseteq \mathcal H_A$ and $\left\{\ket{w_i}\right\}_i \subseteq \mathcal H_B$. By defining $\mathcal H_A^0 := \Supp_A(\ket\psi)^\perp$ we now get the decomposition $\mathcal H_A = \mathcal H_A^0 \dsum \Supp_A(\ket\psi)$. However, now note that as $\left((EF)^+ \tensor I_B\right) \ket\psi = 0$, in particular we have $\left(E^-F^- \tensor I_B\right)\ket\psi = 0$. But as $\left((E^-+F^-)\tensor I_B\right) \ket\psi = \ket\psi$, this implies by Lemma \ref{lem:state-preserve-implies-identity} that on $\Supp_A(\ket\psi)$, $E^-+F^-$ is the identity operator, and therefore, as they are projections, they partition $\Supp_A(\ket\psi)$ into two subspaces, $\mathcal H_A^1 := \Ran(E^-) \cap \Supp_A(\ket\psi)$ and $\mathcal H_A^2 := \Ran(F^-) \cap \Supp_A(\ket\psi)$. Thus, we now have the desired decomposition $\mathcal H_A = \mathcal H_A^0 \oplus \mathcal H_A^1 \oplus \mathcal H_A^2$. 

Similarly, we get a decomposition $\mathcal{H}_B = \mathcal H_B^0 \oplus \mathcal H_B^1 \oplus \mathcal H_B^2$ of Bob's space using the operators $G$ and $H$.

We now wish to show that we can split the state $\ket{\psi}$ into a direct sum, but to do so requires us to show that $\Supp \ket\psi$ is orthogonal to both $\mathcal H_A^1\tensor \mathcal H_B^2$ and $\mathcal H_A^2 \tensor \mathcal H_B^1$, which we proceed to do now.

As the strategy is perfect, it is in particular consistent, i.e., $\squeeze\psi{A_i\tensor B_i} = 1$ for all $i \in [18]$. Hence, we have $A_i \tensor B_i\ket\psi = \ket\psi$. Also, as noted in the second paragraph of this proof, $E(-F)\tensor I_B\ket\psi = \ket\psi$. Thus we find:
\begin{align*}
1 &= \squeeze\psi{(A_3 \tensor B_3)(A_6 \tensor B_6)(A_9 \tensor B_9)(E(-F)\tensor I_B)} \\
&= \squeeze\psi{(-E^2F \tensor G)} = -\squeeze\psi{F \tensor G},
\end{align*}
where we used $E^2 = I_A$. But this implies $(F \tensor G)\ket\psi = -\ket\psi$, and then as $\norm{F\tensor G} \leq 1$, we find $(F\tensor G)^+\ket\psi = 0$, and in particular, $\left(F^- \tensor G^-\right)\ket\psi = 0$. Similarly we find $(E^-\tensor H^-) \ket\psi = 0$, and in total we therefore have
\begin{align*}
\ket\psi &= \left((E^-+F^-)\tensor(G^-+H^-)\right) \ket\psi\\
&= \left(E^-\tensor G^-\right) \ket\psi + \left(E^-\tensor H^-\right)\ket\psi + \left(F^-\tensor G^-\right) \ket\psi + \left(F^-\tensor H^-\right)\ket\psi\\
&= \left(E^-\tensor G^-\right) \ket\psi + \left(F^-\tensor H^-\right)\ket\psi\numberhere\label{eq:gms-selftest-state-decomp}
\end{align*}
By definition of the subspaces, we therefore have $\Supp\ket\psi \perp \mathcal H_A^2\tensor \mathcal H_B^1$ and $\Supp\ket\psi \perp \mathcal H_A^1\tensor \mathcal H_B^2$.

Therefore we can define $\ket{\varphi_1} := \left(E^-\tensor G^-\right)\ket\psi$ and $\ket{\varphi_2} := \left(F^-\tensor H^-\right)\ket\psi$; then we have $\ket\psi = \ket{\varphi_1} \dsum \ket{\varphi_2}$.
Now we can define two sub-strategies $\mathcal S_k$, for $k \in \{1, 2\}$ under the assumption that $\ket{\varphi_k} \neq 0$:
\begin{equation*}
\mathcal S_k = \left(\frac{\ket{\varphi_k}}{\norm{\ket{\varphi_k}}}\in \mathcal H_A^k \tensor \mathcal H_B^k, \left\{A_i|_{\mathcal H_A^k}\right\}_{i\in[18]}, \left\{B_j|_{\mathcal H_A^k}\right\}_{j\in[18]}\right)
\end{equation*}
However, we must be careful and verify that this strategy is well-defined -- in particular to use it we need to apply several operators in succession, so we need to verify that $\Ran A_i|_{\mathcal H_A^k} \subseteq \mathcal H_A^k$ (and similarly for $\mathcal H_B^k$). However, as $\squeeze\psi{\left(A_i\tensor B_i\right)} = 1$, we have $\Ran\left(\left(A_i\tensor B_i\right)|_{\Supp_A(\ket\psi)\tensor \mathcal H_B}\right) \subseteq \Supp_A(\ket\psi) \tensor \mathcal H_B$. This implies that $\Ran A_i|_{\mathcal H_A^1\dsum \mathcal H_A^2} \subseteq \mathcal H_A^1 \dsum \mathcal H_A^2$, and so it is well-defined to restrict observables to this space. Now, further restricting to $\mathcal H_A^1$, we note that this can be accomplished by first applying $E^-|_{\mathcal H_A^1\dsum \mathcal H_A^2}$, i.e.:
\begin{equation*}
A_i|_{\mathcal H_A^1} = A_i|_{\mathcal H_A^1\dsum \mathcal H_A^2}E^-|_{\mathcal H_A^1\dsum \mathcal H_A^2} = A_i|_{\mathcal H_A^1\dsum \mathcal H_A^2}\left(E^-|_{\mathcal H_A^1\dsum \mathcal H_A^2}\right)^2 = E^-|_{\mathcal H_A^1\dsum \mathcal H_A^2}A_i|_{\mathcal H_A^1\dsum \mathcal H_A^2}E^-|_{\mathcal H_A^1\dsum \mathcal H_A^2}
\end{equation*}
Here we have used that $\comm{E^-|_{\mathcal H_A^1\dsum \mathcal H_A^2}, A_i|_{\mathcal H_A^1\dsum \mathcal H_A^2}} = 0$, and thus we can commute $E^-|_{\mathcal H_A^1\dsum \mathcal H_A^2}$ to the beginning of the equation. Observe that while Lemma \ref{lem:glue-ms-commutation} only gives $\comm{E^-|_{\mathcal H_A^1\dsum \mathcal H_A^2}, A_i|_{\mathcal H_A^1\dsum \mathcal H_A^2}} = 0$ for $i \in [9]$, it also gives $\comm{F^-|_{\mathcal H_A^1\dsum \mathcal H_A^2}, A_i|_{\mathcal H_A^1\dsum \mathcal H_A^2}} = 0$ for $i \in \{10,\dots, 18\}$, and Equation \ref{eq:gms-proof-ef-is-identity} implies that $E^-|_{\mathcal H_A^1 \dsum \mathcal H_A^2} + F^-|_{\mathcal H_A^1 \dsum \mathcal H_A^2} = I|_{\mathcal H_A^1 \dsum \mathcal H_A^2}$, so $\comm{E^-|_{\mathcal H_A^1\dsum \mathcal H_A^2}, A_i|_{\mathcal H_A^1\dsum \mathcal H_A^2}} = 0$ as well.
However, this now implies $\Ran\left(A_i|_{\mathcal H_A^1}\right) \subseteq \mathcal H_A^1$, and similar arguments show that the range of the rest of the operators is also contained within the space they are restricted to. This implies that the strategy is indeed well-defined.

However, it is now easy to verify that $\mathcal S_1$ and $\mathcal S_2$ both are perfect strategies for the Magic Square game, and therefore by self-testing of the Magic Square game, Theorem \ref{thm:magic-square-selftest}, for each $k\in\{1, 2\}$, there exists isometries $V_{k,A} : \mathcal H_A^k \to \mathbb C^4 \tensor \mathcal H_{A,\aux}^k$ and $V_{k,B} : \mathcal H_B^k \to \mathbb C^4 \tensor \mathcal H_{B,\aux}^k$ and normalised states $\ket{\aux_k}$ such that
\begin{align*}
\left(V_{k, A} \tensor V_{k, B}\right)\frac{\ket{\varphi_k}}{\norm{\ket{\varphi_k}}} &= \ket{\psi_4} \tensor \ket{\aux_k} && k \in \{1, 2\}\numberhere\label{eq:gms-selftest-proof-state-isometry}\\
\left(V_{1, A} \tensor V_{1, B}\right)\left(A_i|_{\mathcal H_A^1}\tensor B_j|_{\mathcal H_B^1}\right)\frac{\ket{\varphi_1}}{\norm{\ket{\varphi_1}}} &= \left(\tilde A_i \tensor \tilde B_j\right)\ket{\psi_4} \tensor \ket{\aux_1}, && i,j\in[9]\\
\left(V_{2, A} \tensor V_{2, B}\right)\left(A_i|_{\mathcal H_A^2}\tensor B_j|_{\mathcal H_B^2}\right)\frac{\ket{\varphi_2}}{\norm{\ket{\varphi_2}}} &= \left(\tilde A_{i-9} \tensor \tilde B_{j-9}\right)\ket{\psi_4} \tensor \ket{\aux_2}, && i,j\in\{10,\dots, 18\},
\end{align*}
where the operators $\tilde A_i,\tilde B_j$, $i,j\in[9]$ are the ideal operators for the Magic Square game.

We now need only show that the remaining operators constitute a representation of $(\mathbb Z/2\mathbb Z)^{\times 4}$. Therefore we consider how $A_1, \dots, A_9$ acts on $\mathcal H_A^2$; note that how $A_{10}, \dots, A_{18}$ acts on $\mathcal H_A^1$ follows symmetrically. Recall that $A_3A_6A_9A_{10}A_{13}A_{16} = -I$ when restricted to $\mathcal H_A^1 \dsum \mathcal H_A^2$, which in particular implies it holds on $\mathcal H_A^2$. However, that space is characterised by $A_{10}A_{13}A_{16} = -I$, which implies $A_3A_6A_9 = I$ on $\mathcal H_A^2$. However, then we see (note that this is restricted to $\mathcal H_A^2$ -- we temporarily suppress the notation indicating restriction):
\begin{align*}
A_2A_4A_2A_4 = A_3A_1A_1A_7A_8A_5A_5A_6 = A_3\overbrace{A_7A_8}^{\mathclap{=A_9}}A_6 = A_3A_9A_6 = I
\end{align*}
Thus, $\comm{A_2, A_4} = 0$, and symmetrically, $\comm{A_1, A_5} = 0$. We note that we can write
\begin{equation*}
(\mathbb Z/2\mathbb Z)^{\times 4} = \langle x_1, x_2, x_3, x_4 \mid x_i^2, \comm{x_i, x_j}, i,j\in[4]\rangle,
\end{equation*}
so define the function $f : G \to \mathcal B(\mathcal H_A^2)$ by
\begin{align*}
f(x_1) &= A_1|_{\mathcal H_A^2} & f(x_2) &= A_2|_{\mathcal H_A^2}\\
f(x_3) &= A_4|_{\mathcal H_A^2} & f(x_4) &= A_5|_{\mathcal H_A^2},
\end{align*}
we find that the images of the generators satisfy the constraints of the group, and therefore it is well-defined to merely extend $f$ by multiplication (in any way) to cover the entire group. But this then shows that $A_1, \dots, A_9$ restricted to $\mathcal H_A^2$ gives a representation of $(\mathbb Z/2\mathbb Z)^{\times 4}$.

Now define $\sigma_2(g) := V_{2,A}f(g)V_{2,A}^*$, and observe that since \eqref{eq:gms-selftest-proof-state-isometry} holds, we therefore have:
\begin{align*}
\left(V_{2, A} \tensor V_{2, B}\right)(f(g) \tensor I)\frac{\ket{\varphi_2}}{\norm{\ket{\varphi_2}}} &= \left(V_{2, A} \tensor V_{2, B}\right)(f(g) \tensor I)\left(V_{2, A} \tensor V_{2, B}\right)^*\left(V_{2, A} \tensor V_{2, B}\right)\frac{\ket{\varphi_2}}{\norm{\ket{\varphi_2}}}\\
&= (V_{2,A}f(g)V_{2,A}^*)\left(V_{2, A} \tensor V_{2, B}\right)\frac{\ket{\varphi_2}}{\norm{\ket{\varphi_2}}}\\
&= \sigma_2(g)\left(\ket{\psi_4}\tensor \ket{\aux_2}\right),
\end{align*}
and thus we find that the strategy restricted to $\mathcal H_A^2$ is a local dilation of $\mathcal S_2^{\sigma_2}$. Similarly we find that there is a representation $\sigma_1$ of $(\mathbb Z/2\mathbb Z)^{\times 4}$ such that the strategy restricted to $\mathcal H_A^1$ is a local dilation of $\mathcal S_1^{\sigma_1}$.

Finally note that we can remove the assumption that $\ket{\varphi_k}\neq 0$, as in that case, we must have $\ket{\varphi_{3-k}} = \ket\psi$, implying that $\mathcal S$ is then equal to $\mathcal S_{3-k}$ while $\mathcal S_k$ will then be zero. The above then implies that either $\mathcal S_1^\sigma$ or $\mathcal S_2^\tau$ is a local dilation of the nonzero strategy. 
\end{proof}

\begin{corollary}
\label{cor:gms-state-selftest}
The Glued Magic Square game self-tests the state $\ket{\psi_4}$.
\end{corollary}
\begin{proof}
Fix a perfect strategy $\mathcal S$ for the Glued Magic Square, and use Theorem \ref{thm:gms-convex-self-test} to decompose this into $\alpha_1\mathcal S_1^\sigma+\alpha_2\mathcal S_2^\tau$. Let $\ket{\varphi_k}$ denote the state used in $\mathcal S_k$, $k \in [2]$, and note that these are both local dilations of $\mathcal S_{\text{MS}}$ whose state is $\ket{\psi_4}$. This implies that there exists isometries $V_{k,A}$ and $V_{k,B}$, for $k\in[2]$ such that $\left(V_{k,A}\tensor V_{k,B}\right)\ket{\varphi_k} = \ket{\psi_4}\tensor \ket{\aux_k}$ for some auxiliary states $\ket{\aux_k}$.

By defining $V_A := V_{1,A} \dsum V_{2,A}$ and $V_B := V_{1,B} \dsum V_{2,B}$ and extending these to $\mathcal H_A$ and $\mathcal H_B$, respectively, we find, by observing $\ket\psi = \alpha_1\ket{\varphi_1} + \alpha_2\ket{\varphi_2}$, that
\begin{align*}
\left(V_A\tensor V_B\right)\ket\psi &= \left(\left(V_{1,A} \dsum V_{2,A}\right)\tensor \left(V_{1,B} \dsum V_{2,B}\right)\right)\left(\alpha_1\ket{\varphi_1} \oplus \alpha_2\ket{\varphi_2}\right)\\
&= \bigdsum_{k=1}^2\alpha_k\left(V_{k,A} \tensor V_{k,B}\right)\ket{\varphi_k}\\
&= \bigdsum_{k=1}^2\ket{\psi_4}\tensor(\alpha_k\ket{\aux_k})\\
&= \ket{\psi_4} \tensor \left(\alpha_1\ket{\aux_1}\dsum \alpha_2\ket{\aux_2}\right),
\end{align*}
and by defining $\ket\aux := \alpha_1\ket{\aux_1}\dsum \alpha_2\ket{\aux_2}$ and noting that $\norm{\ket{\aux_1}}^2 + \norm{\ket{\aux_2}}^2 = 1$, we get the desired result.
\end{proof}
While it may seem trivial at first to extend this proof to the robust case by using robustness of the constituent games, this is a little more difficult, as our result also relies on the restriction to the eigenspaces of some operators being well-defined. As our proofs of this in Lemma \ref{lem:glue-ms-commutation} and Lemma \ref{lem:glue-mp-commutation} relied on all constraints being satisfied exactly, they cannot immediately be translated into the robust case, however in Section \ref{sec:robustness}, we sketch the required arguments.

\begin{example}
\label{ex:gms-example-strategy}
We have fully characterised all (pure) strategies for the Glued Magic Square game in Theorem \ref{thm:gms-convex-self-test}, so it is worth considering an example of such a strategy.
Let $\alpha, \beta > 0$ be two real numbers which satisfy $\alpha^2 + \beta^2 = 1$, and let $\ket\xi \in \mathbb C^5 \tensor \mathbb C^5$ be some state,. Then, the following state and observables constitute a perfect strategy for the Glued Magic Square game:
\begin{align*}
&&\ket\psi &:= \alpha \ket{\psi_4} \dsum \beta \left(\ket{\psi_4} \tensor \ket\xi\right)\\
&&A_i = B_i &:= \tilde A_i \dsum \left(\vphantom{\big|}\left(I_1 \dsum (-I_1) \dsum I_2\right)\tensor I_5\right) && i \in [9]\\
\phantom{i\in[9]}&&A_{19-i} = B_{19-i} &:= \left(\vphantom{\big|}(-I_2) \dsum I_1 \dsum (-I_1)\right) \dsum \left(\tilde A_i \tensor I_5\right) && i \in [9],
\end{align*}
where $\tilde A_i \in \mathcal B(\mathbb C^4)$ are the optimal operators for the Magic Square game. Note the structure of this strategy: it is composed of two perfect strategies for the Magic Square game -- one for each constituent part of the Glued Magic Square game. However, while the operators $A_1, \dots, A_9$ have to constitute a perfect strategy for the Magic Square game on one part, they don't have to be identity on the other part, as we can take them to be any representation of $\left(\mathbb Z/2\mathbb Z\right)^{\times 4}$ (where $A_1, A_2, A_4$ and $A_5$ corresponds to the natural generators of this group). Therefore it is in particular possible to use the direct sum of several different irreducible representations, which we have done here.
\end{example}

\subsubsection*{Acknowledgements}
L.~Man\v cinska and J.~Prakash are supported by Villum Fonden via the QMATH Centre of Excellence (Grant No. 10059). We thank J\c edrzej Kaniewski for helpful comments on a draft of this paper.

\printbibliography
\appendix

\section{Appendix}
\subsection{Magic Pentagram}
\label{sec:pentagram}
It should be noted that while we have only considered the Glued Magic Square in the main text, the proof of self-testing can easily be extended to the case where either one or both parts are the Magic Pentagram. However, for that, we need to show that gluing the Magic Pentagram with some binary LCS game, the observables used for the Magic Pentagram part will also satisfy some particular commutation relations. Note then that gluing two copies of Magic Pentagram together yields a self-test for the state $\ket{\psi_8}$, while one can only extract $\ket{\psi_4}$ if any part is Magic Square. Of course it should be noted that this self-testing result also requires an application of Theorem \ref{thm:magic-pentagram-selftest}, which states that the Magic Pentagram game self-tests $\ket{\psi_8}$ along with some ideal observables.
\begin{lemma}\label{lem:glue-mp-commutation}
Fix a perfect strategy $\left(\ket\psi \in \mathcal H_A\tensor \mathcal H_B, \{A_x\}_{x}, \{B_y\}_{y}\right)$ for a glued pseudo-telepathic LCS with a Magic Pentagram part, and suppose the observables $A_1, \dots, A_{10}$ for Alice (resp., $B_1, \dots, B_{10}$ for Bob) correspond to the Magic Pentagram part as illustrated in Figure \ref{subfig:magic-pentagram-game}. Then with $E := A_2A_3A_4A_5$ and $F := B_2B_3B_4B_5$, it holds that \begin{align*}
\comm{E, A_i}\tensor I_B\ket\psi = 0, \quad \text{ and } \quad I_A \tensor \comm{F, B_j}\ket\psi = 0
\end{align*}
for all $i,j \in [10]$.\end{lemma}

\begin{proof}
We proceed as in Lemma \ref{lem:glue-ms-commutation}, and assume that the state has full Schmidt rank; therefore the initial observations of Lemma \ref{lem:glue-ms-commutation} also hold here. In particular, every product of operators along a straight line in Figure \ref{subfig:magic-pentagram-game} is identity, and we will make use of this fact repeatedly. Observe initially that $A_2$, $A_3$, $A_4$ and $A_5$ each commute with $E$ (otherwise the strategy wouldn't be well-defined), and therefore we need only prove the statement for $i \in \{1,6,7,8,9,10\}$. We initially find that
\begin{align*}
A_2A_5 &= A_6A_8A_{10}A_7A_8A_9\\
&= A_8^2A_6A_{10}A_7A_9\\
&= A_6(A_{10}A_7)A_9\\
&= A_6(A_1A_4)A_9\\
&= A_3A_9A_4A_9,
\end{align*}
where we in the second equality have used that $A_8$ commutes with all operators present in the expression, and otherwise the relations forcing certain products to be identity (i.e. $A_1A_4A_7A_{10} = I$ and the factors pairwisely commute, so $A_{10}A_7 = A_1A_4$). Now, we therefore have
\begin{align*}
E = A_2A_3A_4A_5 &= A_3A_4A_2A_5\\
&= A_3A_4A_1^2A_2A_5\\
&= A_1A_3A_4 A_1\overbrace{A_2A_5}^{\mathclap{=A_3A_9A_4A_9}}\\
&= A_1A_3A_4\overbrace{A_1A_3A_9}^{=A_6}A_4A_9\\
&= A_4A_9A_4A_9,
\end{align*}
so in particular, as $A_1$ commutes with $A_4$ and $A_9$, $\comm{E, A_1} = 0$. But the same holds true for $A_7$, so $\comm{E, A_7} = 0$. Symmetrically, we can show $E = A_3A_{10}A_3A_{10}$, and since $A_3$ and $A_{10}$ each commute with $A_6$ (as they share a constraint), $\comm{E, A_6} = 0$.

Now, as $A_9 = A_1A_3A_6$, and each factor has been found to commute with $E$, $\comm{E, A_9} = 0$. Symmetrically, $\comm{E, A_{10}} = 0$. Finally, $A_8 = A_5A_7A_9$, each of which commute with $E$, so $\comm{E, A_8} = 0$.
\end{proof}

\subsection{Robustness}
\label{sec:robustness}

In this appendix we sketch the arguments required to show that the Glued Magic Square in fact \emph{robustly} self-tests the state $\ket{\psi_4}$. In particular, while we do not define robust self-testing explicitly, we show that just as in Theorem \ref{thm:gms-convex-self-test} one can from a strategy winning the Glued Magic Square game with probability $1-\varepsilon$ extract strategies winning the Magic Square game with probability $1-O(\varepsilon)$ (note, though, that one of these strategies can be zero just as in the perfect case). By using robustness of the Magic Square one can then approximately extract $\ket{\psi_4}$ from each of the two substrategies.

We start out by proving a few lemmas that will enable us to show this result. Note that these are not necessarily novel -- a result similar to Lemma \ref{lem:approx-cycling} appears in \cite{cs-robust-lcs}.
\begin{lemma}
\label{lem:approx-lrmul}
Suppose $\ket\psi \in \mathcal H$ is a state, and $A,B \in \mathcal B(\mathcal H)$ are unitary operators satisfying $\Re\squeeze\psi A \geq 1 - \varepsilon$ and $\Re\squeeze\psi B \geq 1 - \delta$ for some $\varepsilon, \delta \geq 0$. Then,\begin{align*}
    \Re\squeeze\psi{AB} \geq 1 - \left(\sqrt\varepsilon+\sqrt\delta\right)^2.
\end{align*} Moreover, if $\squeeze\psi A$ and $\squeeze\psi B$ are real, then so is $\squeeze\psi{AB}$.
\end{lemma}
\begin{proof}
Observe that $\norm{(A-I)\ket\psi}^2 = 2-2\Re\squeeze\psi A \leq 2\varepsilon$, and similarly $\norm{(B-I)\ket\psi}^2 \leq 2\delta$. Thus, using Cauchy-Schwarz inequality,
\begin{equation*}
|\Re\squeeze\psi{(A-I)(B-I)}| \leq \left|\squeeze\psi{(A-I)(B-I)}\right|
\leq \norm{(A-I)\ket\psi}\cdot\norm{(B-I)\ket\psi}
\leq 2\sqrt{\varepsilon\delta}
\end{equation*}
Since, $\squeeze\psi{(A-I)(B-I)} = \squeeze\psi{AB}+\squeeze\psi I-\squeeze\psi A-\squeeze\psi B,$ upon rearranging and taking the real value yields
\begin{align*}
\Re\squeeze\psi{AB} &= \Re\squeeze\psi A + \Re\squeeze\psi B - 1 + \Re\squeeze\psi{(A-I)(B-I)} \\
&\geq 1 - \varepsilon + 1 - \delta - 1 - 2\sqrt{\varepsilon\delta} = 1 - \left(\sqrt\varepsilon+\sqrt\delta\right)^2,
\end{align*} as required.
\end{proof}

Observe that if $U_1, \dots, U_n$ are all unitaries satisfying $\Re\squeeze\psi{U_i} \geq 1 - \varepsilon$, then repeatedly applying Lemma \ref{lem:approx-lrmul} we obtain $\Re\squeeze\psi{U_1U_2\cdots U_n} \geq 1-n^2\varepsilon$. In general, one can also note that if $\Re\squeeze\psi{U_i} \geq 1 - C_i\varepsilon$ for some constants $C_i \geq 0$, then there is a constant $C \geq 0$ such that $\Re\squeeze\psi{U_1U_2\cdots U_n} \geq 1-C\varepsilon$.

\begin{lemma}
\label{lem:approx-cycling}
Suppose $\ket\psi\in\mathcal H_A \tensor \mathcal H_B$ is a state and $A, U \in \mathcal B(\mathcal H_A)$ and $B \in \mathcal B(\mathcal H_B)$ are all unitary operators. Furthermore suppose there exists some $\varepsilon \geq 0$ such that $\squeeze\psi{(U\tensor I)} \geq 1 - \varepsilon$ and $\squeeze\psi{(A\tensor B)} \geq 1 - \varepsilon$. Then it holds that $\squeeze\psi{(AUA^*\tensor I)} \geq 1-9\varepsilon$.
\end{lemma}

\begin{proof}
Observe that $\squeeze\psi{(A\tensor B)} = \squeeze\psi{(A\tensor B)^*}$ and $
\squeeze\psi{(AUA^* \tensor I)} = \squeeze\psi{(A\tensor B)(U\tensor I)(A\tensor B)^*},$ so applying Lemma \ref{lem:approx-lrmul} twice yields the desired conclusion.
\end{proof}
Note that as before, if $\squeeze\psi{(U\tensor I)} \geq 1 - O(\varepsilon)$ and $\squeeze\psi{(A\tensor B)} \geq 1 - O(\varepsilon)$, then this lemma implies $\squeeze\psi{(AUA^*\tensor I)} \geq 1-O(\varepsilon)$.

These two results forms the main ingredients in proving robustness. In the ideal case, we were able to prove that restricting operators to certain eigenspaces was well-defined, but that is somewhat more difficult to do robustly. We will need a way to translate approximate commutation for our entire quantum strategy into an approximate commutation on some eigenspace. However, to do this we first need to show that if we restrict to some space then approximate relations are still conserved.
\shinynew{
\begin{lemma}
\label{lem:approx-identity-survives-decomposition}
Suppose $\ket\psi \in \mathcal H_A \tensor \mathcal H_B$ is a unit vector, $A \in \mathcal B(\mathcal H_A)$ is a self-adjoint operator of norm at most 1, and $\squeeze\psi{(A\tensor I)} \geq 1 - \varepsilon$ for some $\varepsilon \geq 0$. Then for any decomposition $\ket\psi = \ket\varphi \dsum \ket\phi$ it holds that $\squeeze\varphi{(A\tensor I)} \geq \norm{\ket\varphi}^2 - \varepsilon$.
\end{lemma}
\begin{proof}
Schmidt decompose $\ket\psi = \sum_i\lambda_i\ket{v_i}\ket{w_i}$ with $\lambda_i>0$ for all $i$, and furthermore for each $i$, find $c_i \geq 0$ such that $1-c_i\varepsilon = \squeeze{v_i}A$ (this is always possible as the right-hand side is bounded by 1). Then we find
\begin{equation*}
1-\varepsilon \leq \squeeze\psi{(A\tensor I)} = \sum_{i,j}\lambda_i\lambda_j\bra{v_j}A\ket{v_i}\bra{w_j}I\ket{w_i} = \sum_i\lambda_i^2\squeeze{v_i}A = \sum_i\lambda_i^2\left(1-c_i\varepsilon\right),
\end{equation*}
and using that $\sum_i\lambda_i^2 = 1$, we can subtract this from the above to obtain $\sum_i\lambda_i^2c_i \leq 1$.

Now as $\ket\psi = \ket\varphi \dsum \ket\phi$, there exists coefficients $\{\alpha_i\}_i \subseteq [0, 1]$ such that $\ket\varphi = \sum_i\alpha_i\lambda_i\ket{v_i}\ket{w_i}$ (these can be obtained by projecting $\lambda_i\ket{v_i}\ket{w_i}$ onto the direct summand in which $\ket\varphi$ lies). By orthonormality one has $\norm{\ket\varphi}^2 = \sum_i\alpha_i^2\lambda_i^2$, using which we find
\begin{equation*}
\squeeze\varphi A = \sum_i\alpha_i^2\lambda_i^2\squeeze{v_i} A = \sum_i \alpha_i^2\lambda_i^2(1-c_i\varepsilon) = \norm{\ket\varphi}^2 - \varepsilon\sum_i\alpha_i^2\lambda_i^2c_i\geq \norm{\ket\varphi}^2 - \varepsilon,
\end{equation*}
where we have used that $\alpha_i^2 \leq 1$ to bound $\sum_i\alpha_i^2\lambda_i^2c_i \leq \sum_i\lambda_i^2c_i \leq 1$.
\end{proof}
}

Now recall that our proof of self-testing in the exact case relied upon restriction to certain eigenspaces being well-defined. That is not true in the approximate case, but we can now show that if all observables approximately commute with some specific observable, then so do they do with its $-1$-eigenprojector.

\shinynew{
\begin{lemma}
\label{lem:approx-commutation-survives-decomposition}
Suppose $\ket\psi \in \mathcal H_A \tensor \mathcal H_B$ is a unit vector, $A, E \in \mathcal B(\mathcal H_A)$ are operators of norm at most 1, that $A$ is an observable and $E$'s only eigenvalues are $\pm1$. Suppose that $A$ and $E$ approximately commute on $\ket\psi$ in the sense that $\squeeze\psi{(AEA^*E^*\tensor I)} \geq 1 - \varepsilon$ for some $\varepsilon \geq 0$. Then $\squeeze\varphi{(AE^-A^*(E^-)^*\tensor I)} \geq 1-\frac\varepsilon{\norm{(E^-\tensor I)\ket\psi}^2}$, where $\ket\varphi = \frac{(E^-\tensor I)\ket\psi}{\norm{(E^-\tensor I)\ket\psi}}$.
\end{lemma}
\begin{proof}
As a shorthand, define $\ket\phi := (E^-\tensor I)\ket\psi$, and note that $\ket\varphi = \frac{\ket\phi}{\norm{\ket\phi}}$. By Lemma \ref{lem:approx-identity-survives-decomposition} it holds that $\squeeze\varphi{(AEA^*E^*\tensor I)} \geq \frac{\norm{\ket\phi}^2 - \varepsilon}{\norm{\ket\phi}^2} = 1 - \frac1{\norm{\ket\phi}^2}\varepsilon$.
By using that $E = E^+-E^-$ and that $E^++E^- = I$, we find that $E = I-2E^-$, and therefore we get (by also using that $A$ is an observable and thus $AA^* = I$):
\begin{align*}
1-\frac\varepsilon{\norm{\ket\phi}^2} &\leq \squeeze\varphi{(AEA^*E^*\tensor I)}\\
&= \squeeze\varphi{(AIA^*E^*\tensor I)} - 2\squeeze\varphi{(AE^-A^*E^*\tensor I)}\\
&= 1-2\squeeze\varphi{((E^-)^*\tensor I)}-2\squeeze\varphi{(AE^-A^*\tensor I)} + 4\squeeze\varphi{(AE^-A^*(E^-)^*\tensor I)},\numberhere\label{eq:approx-commutation-proof-needs-bound}
\end{align*}
and noting that $E^-$ is an orthogonal projector (so $(E^-)^2 = E^- = (E^-)^*$), we find that $\squeeze\varphi{((E^-)^*\tensor I)} = 1$. Furthermore, we can write $(AEA^*E^*\tensor I)\ket\psi = \ket\psi - \ket\eta$, where by assumption this is greater than $1-\varepsilon$ when left-multiplying by $\bra\psi$, so we find $0 \leq \braket{\psi|\eta} \leq \varepsilon$. However, as $E^-$ is a projection, it therefore holds that $0 \leq \bra\psi (E^- \tensor I)\ket\eta \leq \varepsilon$. Combining all this, we thus find (using that $E^- = (E^-)^*$):
\begin{align*}
\squeeze\varphi{(AE^-A^*\tensor I)} &= \frac1{\norm{\ket\phi}^2}\squeeze\psi{(E^-AE^-A^*E^-\tensor I)}\\
&= \frac1{\norm{\ket\phi}^2}\bra\psi (E^- \tensor I)\left(\ket\psi - \ket\eta\right)\\
&= \frac1{\norm{\ket\phi}^2}\left(\bra\psi (E^- \tensor I)\ket\psi - \braket{\psi|\eta}\right)\\
&\geq 1-\frac\varepsilon{\norm{\ket\phi}^2}.
\end{align*}
Thus, we can rewrite \eqref{eq:approx-commutation-proof-needs-bound} and use the estimates to get a bound on our desired quantity:
\begin{align*}
4\squeeze\varphi{(AE^-A^*(E^-)^*\tensor I)} &\geq 2\squeeze\varphi{((E^-)^*\tensor I)} + 2\squeeze\varphi{(AE^-A^*\tensor I)} - \frac\varepsilon{\norm{\ket\phi}^2}\\
&\geq 2\cdot 1 + 2\cdot\left(1-\frac\varepsilon{\norm{\ket\phi}^2}\right) - \frac\varepsilon{\norm{\ket\phi}^2}\\
&= 4-3\frac\varepsilon{\norm{\ket\phi}^2},
\end{align*}
and dividing by 4, we obtain a slightly tighter bound than required.
\end{proof}
}

We are now ready to prove approximate commutation of $A_1, \dots, A_9$ with $E = A_3A_6A_9$ for a $(1-\varepsilon)$-optimal strategy for the Glued Magic Square game, where we will then later use Lemma \ref{lem:approx-identity-survives-decomposition} to show that it does in fact restrict to a $(1-O(\varepsilon)$-optimal strategy for the Magic Square game.
\begin{lemma}
\label{lem:glue-ms-approx-commutation}
Fix a strategy $\left(\ket\psi \in \mathcal H_A\tensor \mathcal H_B, \{A_x\}_{x}, \{B_y\}_{y}\right)$ for a glued pseudo-telepathy LCS with a Magic Square part achieving a winning probability of $1-\varepsilon$ for some $\varepsilon > 0$, and suppose the observables $A_1, \dots, A_9$ for Alice and $B_1, \dots, B_9$ correspond to the Magic Square part, as illustrated in Figure \ref{subfig:magic-square-game}. Then it holds that for $E := A_3A_6A_9$ and $F:=B_3B_6B_9$, we have
\begin{align*}
\squeeze\psi{(EA_iE^*A_i^*\tensor I_B)} = 1-O(\varepsilon), \quad \text{ and } \quad \squeeze\psi{(I_A \tensor FB_jF^*B_j^*)} = 1-O(\varepsilon)
\end{align*}
for all $i,j \in [9]$.
\end{lemma}
\begin{proof}
The proof follows in mostly the same way as in that of Lemma \ref{lem:glue-ms-commutation}, except that we obviously cannot conclude that certain equalities hold. Instead, we will make use of Lemma \ref{lem:approx-lrmul} and \ref{lem:approx-cycling} to insert approximately satisfied relations and show that the resulting expression is still close to 1. Therefore, note that as Alice can get 11 possible questions and wins using her strategy along with Bob with probability $1-\varepsilon$, she correctly answers every question with probability at least $1-11\varepsilon$. This implies, for example, that $\squeeze\psi{(A_1A_2A_3 \tensor I)} = 1 - O(\varepsilon)$, and similarly for other products of operators along rows and columns, except for the column corresponding to the odd constraint. We therefore find:
\begin{align*}
1 = \squeeze\psi{I} &= \squeeze\psi{(A_2A_2\tensor I)}\\
&= \squeeze\psi{((A_2A_7A_7A_2)\tensor I)}\\
&= \squeeze\psi{((A_2A_7A_2A_2A_7A_2)\tensor I)}\\
&= \squeeze\psi{((A_2(A_1A_4)(A_5A_8)A_2(A_1A_4)(A_5A_8))\tensor I)}-O(\varepsilon)\\
&= \squeeze\psi{((A_3A_6A_8A_3A_6A_8)\tensor I)}-O(\varepsilon)\\
&= \squeeze\psi{((A_2A_1\;A_4A_5\;A_8\;A_2A_1\;A_4A_5\;A_8)\tensor I)}-O(\varepsilon)\\
&= \squeeze\psi{((A_3A_6A_9A_8A_3A_6A_9A_8)\tensor I)} - O(\varepsilon)\\
&= \squeeze\psi{((EA_8EA_8)\tensor I)} - O(\varepsilon)
\end{align*}
Note that in the first three lines we have merely used that $A_2$ and $A_7$ are observables, while we have inserted several relations using Lemma \ref{lem:approx-lrmul} between the third and fourth line. As Lemma \ref{lem:approx-lrmul} does not directly allow insertion inside a product but only at the two ends, we have used Lemma \ref{lem:approx-cycling} to cyclically shift the expression such that the place we want to insert a relation is at the beginning or end. 

Now note that the above expression is indeed what we want, as both $E$ and $A_8$ are self-adjoint. The proof of approximate commutation for $A_i$ with $E$ for $i \in \{1,2,4,5,7\}$ follow in a similar way (for $i\in\{3,6,9\}$ it holds exactly as this is required for the strategy to be well-defined), and so does the approximate commutation between $B_j$ and $F$, for $j \in [9]$.
\end{proof}

\shinynew{
\begin{theorem}
Suppose $\mathcal S = \left(\ket\psi \in \mathcal H_A\tensor \mathcal H_B, \left\{A_i^{(x)}\;\middle|\; x \in X\right\}_{i\in A},\left\{B_j\right\}_{j\in B}\right)$ is a strategy winning the Glued Magic Square with probability $1-\varepsilon$ for some $\varepsilon \geq 0$. For each $i \in [9]$, fix an equation $x$ where the observable $A_i^{(x)}$ is used, and let $A_i := A_i^{(x)}$.
Then by defining $E = A_3A_6A_9$, $G = B_3B_6B_9$ and $\ket{\varphi_1} := \frac{(E^-\tensor G^-)\ket\psi}{\norm{(E^-\tensor G^-)\ket\psi}}$, there exists a constant $C \geq 0$ such that the strategy $\mathcal S_1 = \left(\ket{\varphi_1}, \{E^-A_xE^-\}_{x\in[9]}, \{G^-B_yG^-\}_{y\in[9]}\right)$ wins the Magic Square game with probability $1-\frac C{\norm{\ket{\varphi_1}}^2}\varepsilon$. Similarly, by defining $F = A_{10}A_{13}A_{16}$ and $H = B_{10}B_{13}B_{16}$ and $\ket{\varphi_2} := \frac{(F^-\tensor H^-)\ket\psi}{\norm{(F^-\tensor H^-)\ket\psi}}$, there is some constant $C' \geq 0$ such that $\mathcal S_2 = \left(\ket{\varphi_2}, \{F^-A_xF^-\}_{x\in[9]}, \{H^-B_yH^-\}_{y\in[9]}\right)$ wins the Magic Square game with probability $1-\frac{C'}{\norm{\ket{\varphi_1}}^2}\varepsilon$.
\end{theorem}
\begin{proof}
Initially we observe that while Alice may use different operators for the same variable when measuring it in different equations, just fixing one of these observables for each variable yields a set of observables $O(\varepsilon)$-approximately satisfying the relations of the Glued Magic Square game, as per the proof of Lemma \ref{lem:glue-ms-approx-commutation}.

Initially consider the state $\ket\eta = (E^-\tensor I)\ket\psi$, which is in general not normalised, and consider
\begin{align*}
\squeeze\eta{(E^-\tensor G^-)} &= \squeeze\eta{(E^-\tensor I)} - \squeeze\eta{(E^- \tensor G^+)}\\
&= \squeeze\psi{(E^-\tensor I)^3} - \squeeze\psi{(E^-\tensor I)(E^-\tensor G^+)(E^-\tensor I)}\\
&= \squeeze\psi{(E^-\tensor I)} - \squeeze\psi{(E^-\tensor G^+)}\\
&= \norm{(E^-\tensor I)\ket\psi}^2 - \squeeze\psi{(E^-\tensor G^+)}.
\end{align*}
Now note that $\squeeze\psi{(E^-\tensor G^+)}$ gives the probability that Alice's and Bob's are not consistent when Alice is asked for the odd constraint and Bob is asked to assign values to either of the topmost three variables. However, they are consistent with probability $1-O(\varepsilon)$, so we find $\squeeze\psi{(E^-\tensor G^+)} = O(\varepsilon)$. Therefore, by defining the normalised state $\ket\phi = \frac{\ket\eta}{\norm{\ket\eta}}$ and using that $(E^-\tensor I)\ket\eta = \ket\eta$, along with $G^- = I-G^+$, we find by the above that
\begin{equation}
\label{eq:gms-robust-projector-size}
\squeeze{\phi}{(E^-\tensor G^-)} = \frac1{\norm{\ket\eta}^2}\left(\norm{(E^-\tensor I)\ket\eta}^2 - \squeeze\psi{(E^-\tensor G^+)}\right) = 1 - O\left(\frac\varepsilon{\norm{\ket\eta}^2}\right),
\end{equation}
which shows that $E^-\tensor G^-$ approximately acts like the identity operator on $(E^-\tensor I)\ket\eta$.

Note that for any $i \in [9]$ it holds by Lemma \ref{lem:glue-ms-approx-commutation} that $\squeeze\phi{((A_iEA_i^*E^*)\tensor I)} = 1 - O\left(\frac\varepsilon{\norm{\ket\eta}^2}\right)$. We can use this to show that $\mathcal S_1$ approximately satisfies the relations required by the Magic Square by initially using Lemma \ref{lem:approx-identity-survives-decomposition} to conclude that the relation approximately holding on $\ket\psi$ also does so on $\ket\phi$ (note that here $\ket\phi$ is normalised, while Lemma \ref{lem:approx-identity-survives-decomposition} is formulated without that assumption, so we perform a rescaling).
\begin{align*}
1-O(\varepsilon) &= \squeeze\psi{((A_1A_2A_3)\tensor I)}\\
&= \squeeze\phi{((A_1A_2A_3)\tensor I)} - O\left(\frac\varepsilon{\norm{\ket\eta}^2}\right)\\
&= \squeeze\phi{((E^-\tensor I)^6(A_1A_2A_3)\tensor I)} - O\left(\frac\varepsilon{\norm{\ket\eta}^2}\right)\\
&= \squeeze\phi{((E^-A_1E^-E^-A_2E^-E^-A_3E^-)\tensor I)} - O\left(\frac\varepsilon{\norm{\ket\eta}^2}\right)
\end{align*}
Note that in the third equality we have merely used that $E^-\tensor I$ is idempotent, and that it preserves $\ket\phi$. In the fourth equality we have used Lemma \ref{lem:approx-lrmul} and \ref{lem:approx-cycling} to insert the commutation relations between $A_i$ and $E^-$ which we get from Lemma \ref{lem:glue-ms-approx-commutation} and use Lemma \ref{lem:approx-commutation-survives-decomposition} to restrict to $\ket\phi$.

Similarly one can show that the operators $E^-A_iE^-$ all $O\left(\varepsilon/\norm{\ket\eta}^2\right)$-approximately satisfy the relations required of the Magic Square game, when applying them to the state $\ket\phi$, so one may be tempted to end here. However $\mathcal S_1$ uses another state, so we need to translate between those two.

Therefore, suppose $R$ is a relation of the Magic Square game which is approximately satisfied for $\ket\phi$ (i.e. a product of observables approximately preserving this state, for example we could have $R = A_1A_2A_3$ as we discussed above), i.e. that $\squeeze\phi R = 1-O\left(\varepsilon/\norm{\ket\eta}^2\right)$. For ease of notation, let $\alpha := \norm{\ket\eta}^2$ and $\beta := \norm{(E^-\tensor G^-)\ket\psi}^2$; then it holds that $\ket{\varphi_1} = \frac{(E^-\tensor G^-)\ket\psi}{\sqrt\beta}$ and the assumption is that $\squeeze\phi R = 1-O(\varepsilon/\alpha)$. Then by recalling from \eqref{eq:gms-robust-projector-size} that $\squeeze\phi{E^-\tensor G^-} = 1-O\left(\varepsilon/\alpha\right)$, we find using Lemma \ref{lem:approx-lrmul} that
\begin{align*}
1-O\left(\varepsilon/\alpha\right) &= \squeeze\phi{(E^-\tensor G^-)R(E^-\tensor G^-)}\\
&= \frac1\alpha \cdot {\squeeze\psi{(E^-\tensor I)(E^-\tensor G^-)R(E^-\tensor G^-)(E^-\tensor I)}}\\
&= \frac1\alpha \cdot {\squeeze\psi{(E^-\tensor G^-)R(E^-\tensor G^-)}}\\
&= \frac\beta\alpha\squeeze{\varphi_1}R,
\end{align*}
and thus we have $\squeeze{\varphi_1}R = \alpha/\beta - O(\varepsilon/\beta)$. However by using that the square of the norm of a vector is its inner product with itself, we find:
\begin{equation*}
\frac{\alpha}{\beta}-1 = \frac{\alpha-\beta}\beta = \frac{\squeeze\psi{(E^-\tensor(G^--I))}}{\beta} = \frac{\squeeze\psi{(E^-\tensor G^+)}}{\beta} = O\left(\varepsilon/\beta\right),
\end{equation*}
where we used the same observation as just before \eqref{eq:gms-robust-projector-size} to conclude that $\squeeze\psi{(E^-\tensor G^+)} = O(\varepsilon)$. Combining this with the previous we in total find $\squeeze{\varphi_1}R = 1-O(\varepsilon/\beta) = 1-O(\varepsilon/\norm{(E^-\tensor G^-)\ket\psi}^2)$. Therefore we get that $\mathcal S_1$ is an $\left(1-O\left(\varepsilon/\norm{(E^-\tensor G^-)\ket\psi}^2\right)\right)$-optimal strategy for the Magic Square game, since we can perform similar proofs for Bob's strategy. Furthermore, the same proofs go through for $\mathcal S_2$, as long as $(F^-\tensor H^-) \ket\psi \neq 0$.
\end{proof}
}

\end{document}